\documentclass[10pt,journal,comsoc]{IEEEtran}  

\usepackage{cite}
\usepackage{amsmath,amssymb,amsfonts}
\usepackage{graphicx}
\usepackage{graphics}
\usepackage{subfigure}
\usepackage{algorithm}
\usepackage{algorithmic}
\usepackage{multirow}
\usepackage{longtable}
\usepackage{array}
\usepackage{textcomp}
\usepackage{xcolor}
\usepackage{multicol}

\usepackage{amsthm}

\newtheorem{thm}{Theorem}

 \usepackage{booktabs} 

\usepackage{microtype}
\usepackage{subcaption}

\def\BibTeX{{\rm B\kern-.05em{\sc i\kern-.025em b}\kern-.08em
		T\kern-.1667em\lower.7ex\hbox{E}\kern-.125emX}}

\begin{document}

\title{ Adaptive  Multi-UAV Relay Deployment Framework in Satellite Aerial Ground Integrated Systems}

\author{\IEEEauthorblockN{Bhola,~\IEEEmembership{Member,~IEEE}, 
Yu-Jia Chen,~\IEEEmembership{Senior Member,~IEEE},
Ashutosh Balakrishnan,~\IEEEmembership{Member,~IEEE},
Swades De,~\IEEEmembership{Senior Member,~IEEE}, and
Li-Chun Wang,~\IEEEmembership{Fellow,~IEEE}
}

\thanks{Copyright (c) 2026 IEEE. Personal use of this material is permitted. However, permission to use this material for any other purposes must be obtained from the IEEE by sending a request to pubs-permissions@ieee.org.}

\thanks{This work has been partially funded by the National Science and Technology Council under the Grants 114-2221-E-A49 -185 -MY3, and 113-2218-E-A49 -027 -, and 114-2224-E-A49 -002 -, and 114-2218-E-A49 -019 -Taiwan.}

\thanks{
An early version of this work was presented at the IEEE Vehicular Technology Conference (VTC-Fall) 2024~\cite{10757512}.


Bhola is with the Department of Electrical and Computer Engineering, National Yang Ming Chiao Tung University, Taiwan R.O.C. (email: bhola.ee06@nycu.edu.tw) 

Yu-Jia Chen is with the National Central University, Chung-Li 32001, Taiwan (e-mail: yjchen@ce.ncu.edu.tw)

A. Balakrishnan is with the Department of Computer Science and Networks, Télécom Paris, 91120 Paris, France. (email: balakrishnan@telecom-paris.fr)

S. De is with the Department of Electrical Engineering and Bharti School of Telecommunications, Indian Institute of Technology Delhi, India. (email: swadesd@ee.iitd.ac.in)

L.-C. Wang is with the Department of Electrical and Computer Engineering, National Yang Ming Chiao Tung University, Taiwan R.O.C. (email:wang@nycu.edu.tw) \emph{(Corresponding author: Li-Chun Wang.)} 

}

}

\markboth{Preprint for IEEE Transactions on Vehicular Technology}
{}

\maketitle

\begin{abstract}
    The sixth-generation (\(\mathbf{6}\)G) communication networks are expected to provide high data rates, ultra-reliable communication, and massive connectivity, especially in challenging environments such as dense urban areas and disaster-affected regions. However, traditional terrestrial-only networks face significant challenges in these scenarios, including signal blockages from high-rise buildings, traffic congestion, and dynamic user distributions. To address these limitations, we propose the adaptive multi-UAV deployment (AMUD) framework within satellite air-ground integrated networks (SAGINs). The AMUD framework dynamically deploys amplify-and-forward multiple unmanned aerial vehicle relay (UAVr) in with low Earth orbit (LEO) satellites to improve coverage, alleviate congestion, and ensure reliable communication in non-line-of-sight and high-demand conditions. We formulate an optimization problem that aims to jointly maximize the energy efficiency of the total network and the total capacity while ensuring the fairness of the total capacity and satisfying the user's requirements. The simulation results demonstrate that AMUD improves the total capacity of the network, improves the total energy efficiency, and increases the fairness of the capacity compared to traditional LEO satellite and ground base station (LEO-GBS) only systems.
\end{abstract}

\begin{IEEEkeywords}
    $6$G Networks, LEO satellites, multiple UAVr, SAGIN, cooperative diversity, demand-supply aware balance.
\end{IEEEkeywords}

\begin{table}[!t]
\caption{Key Notations}
\label{tab:key_notations}
\centering
\renewcommand{\arraystretch}{1.1}
\begin{tabular}{p{0.98cm} p{6.7cm}}
\toprule
\textbf{Symbol} & \textbf{Definition} \\
\midrule

$\mathbf{u}_i$ & $2$D location vector of user $i$ at time $t$ \\
$\mathbf{U}_j$ & $3$D location vector of UAVr $j$ at time $t$ \\
$\mathbf{h}_{i \leftarrow s}$ & Channel vector from the satellite to user $i$ ($L$ antennas) \\
$\hbar_{j \leftarrow s}$ & Scalar channel coefficient from the satellite to UAVr $j$ \\
{$\mathbf{h}_{i \leftarrow j \leftarrow s}$} & {Channel vector from UAVr $j$ to user $i$ ($L$ antennas)} \\
$x_{\text{sym}}$ & Symbol transmitted by the satellite \\
$P^{\text{tx}}_s$ & Satellite transmit power \\
$p_{i \leftarrow j}$ & Transmit power from UAVr $j$ to user $i$ \\
$\mathcal{G}$ & Amplify-and-forward gain at the UAVr \\
$\mathbf{n}_{i \leftarrow s}$ & AWGN vector at user $i$ from the satellite link \\
$n_{j \leftarrow s}$ & Scalar AWGN at UAVr $j$ from the satellite \\
{$\mathbf{n}_{i \leftarrow j \leftarrow s}$} & {AWGN vector at user $i$ from the UAVr link} \\
$\mathbf{r}^{i \leftarrow s}$ & Received signal vector at user $i$ from the satellite \\
$r^{j \leftarrow s}$ & Received signal at UAVr $j$ from the satellite \\
$\mathbf{r}^{i \leftarrow j \leftarrow s}$ & Received signal vector from UAVr $j$ to user $i$ \\
$\gamma_{i \leftarrow j}$ & SINR of the UAVr-to-user link \\
$\gamma_{j \leftarrow s}$ & SINR of the satellite-to-UAVr link \\
$\gamma_{i \leftarrow s}$ & SINR of the satellite-to-user link \\
$V_{j \leftarrow s}$ & Satellite visibility indicator for UAVr $j$ \\
$p^{\text{Hov}}_j$ & Hovering power consumption of UAVr $j$ \\
$C_j$ & Data rate of the UAVr-assisted link \\
$C_G$ & Data rate of the GBS link \\
$C_{\text{Tot}}$ & Total system network capacity \\
{$ p^{\text{Total}}$} & {Total power consumption of the downlink system}\\
$E_{\text{Tot}}$ & System energy efficiency \\
\bottomrule
\end{tabular}
\end{table}

\section{Introduction}
\label{sec:introduction_3}

\IEEEPARstart{T}{he} demand for reliable wireless access has grown significantly, driven by emerging technologies such as augmented reality, the Internet of Things (IoT), and autonomous vehicles. By 2023, Cisco projects that mobile users will reach 13.1 billion globally, with 29.3 billion devices enabled by the Internet \cite{liu2022evolution}. This surge in demand and increasing quality of service (QoS) requirements lead to dynamic user traffic and hotspots that vary in space-time within wireless networks \cite{9383780}.
Although ground-based stations (GBS) have traditionally served network users, they struggle with increasing user density and signal blockages, especially in urban environments. These random non-line-of-sight (NLOS) conditions and obstructions degrade service quality, creating the need for more agile network solutions.

To address these limitations, 6G communication networks are expected to integrate terrestrial and non-terrestrial elements (i.e., aerial and space-based technologies) to enhance network capacity and resilience~\cite{10757512, 8626457,9275613}. Satellite-air-ground integrated networks (SAGINs) are considered a key enabler of scalable, dynamic, and adaptive 6G systems \cite{AB-ICC2024}.
This paper introduces an AMUD framework within SAGINs that dynamically deploys unmanned aerial vehicle relay (UAVr) to improve coverage and signal-to-interference-plus-noise ratio (SINR), especially for users facing signal blockages or GBS congestion. The framework aligns with the $6$G vision by addressing urban communication bottlenecks and enabling flexible, aerial-ground cooperative relay strategies.

\begin{figure*} [!t]
	\centering
	\includegraphics[width=.85\textwidth]{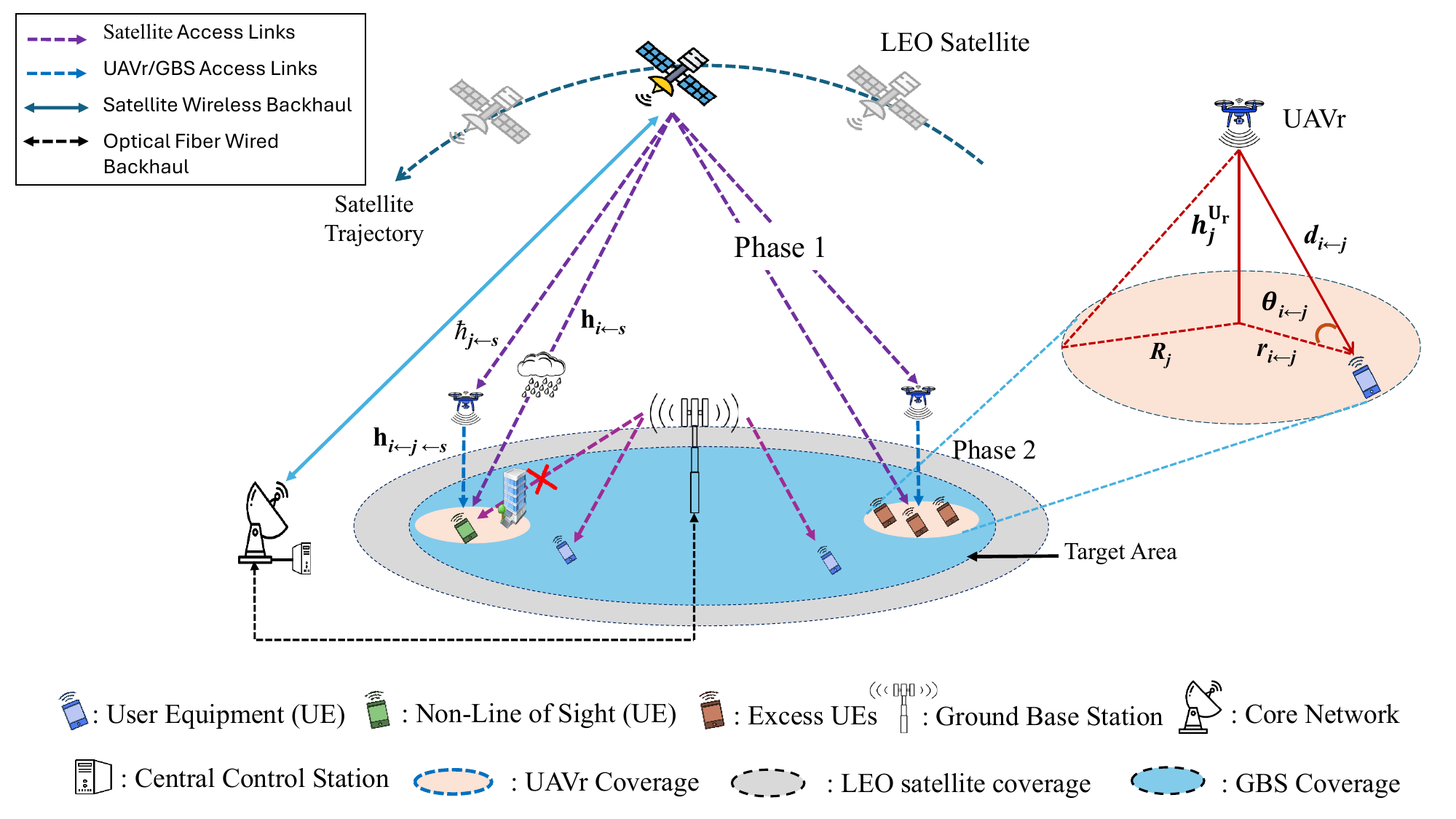}
	\caption{\small {System Model: Hybrid Network Overview.}}
	\label{fig: system model_3}
\end{figure*}

\subsection{Motivation}
Recent interest has focused on integrating terrestrial networks with space-based LEO satellites to provide seamless connectivity and high-speed broadband access for 6G communications. LEO satellite networks consist of a mega-constellation of satellites and wireless backbones. Traditionally, satellites have been used primarily for communication in rural areas. However, in urban environments, satellite links are susceptible to masking effects caused by climatic conditions (e.g., rain, fog) and terrestrial obstructions (e.g., tunnels, dense buildings), which significantly attenuate signals reaching ground users~\cite{9755995}. Furthermore, satellite links, especially those that involve LEO satellites, face challenges in adapting to dynamic network conditions due to their inherent mobility and rapid temporal variations in Internet usage patterns \cite{10.1109/COMST.2022.3197695}. Although LEO satellites experience lower latency compared to medium-earth orbit (MEO) and geostationary earth orbit (GEO) satellites, their shorter orbital paths require continuous tracking of satellite movement related to the ground area under observation \cite{AB-GC2025,3GPPTS22125,8917591,zhao2018exact}

Despite these urban propagation challenges, LEO satellites are increasingly being considered for urban coverage to support seamless global access, offload congested terrestrial networks, and maintain service continuity during infrastructure outages or disasters. However, their practical utility in urban environments is critically dependent on solving the limitations of NLoS imposed by urban obstructions. To this end, UAVr serves as an agile intermediary that can dynamically bridge the NLoS gap between satellites and users, ensuring robust and adaptive urban connectivity.
The deployment of unmanned aerial vehicles as base stations (UAV-BS) has been a conventional method to increase capacity in terrestrial networks \cite{AB-VTC2024, 9946428,9535285,9560149,9555387}. In this paper, we propose the deployment of low-altitude UAVr to address the challenges associated with integrating LEO satellites and improving user QoS. The proposed strategy of deploying multiple UAVr is more energy efficient than traditional UAV-BS while effectively meeting user QoS requirements. Additionally, unlike terrestrial relays in heterogeneous networks, multiple UAVr offer mobility to the mobile operator. This capability allows multiple UAVr to be deployed in rural or disaster scenarios, such as floods, where traditional infrastructure may be compromised.

This paper proposes an adaptive multi-UAVr deployment framework (AMUD) for air-ground satellite integrated networks. In a fixed terrestrial area under observation, particularly with space-time-varying traffic hotspots, the AMUD framework first computes optimal locations for deploying multiple UAVr based on traffic demand distribution. The framework improves network capacity and energy efficiency by improving user QoS. User QoS is improved by exploiting spatial diversity through a cooperative diversity (CD) strategy at the ground user level, which combines signals from separate communication links from LEO satellites and multiple UAVr. The AMUD framework intelligently controls multiple UAVr and signals to provide fair communication opportunities to users within the designated experimental terrestrial coverage area.
%
%
The key contributions of this work are summarized as follows:
\begin{itemize} 

    \item This study presents an innovative architecture that establishes a multilayered communication network that seamlessly integrates multiple UAVr, LEO satellites, and GBS. This configuration enables flexible and dynamic coverage, effectively addressing congestion challenges and enhancing network adaptability. The deployment of multiple UAVr provides localized coverage in regions with varying user densities, while LEO satellites offer extensive coverage and high data transmission rates. The architecture explicitly incorporates satellite visibility and fairness considerations to ensure reliable communication in urban and remote environments.

    \item We propose the AMUD framework, a pioneering approach that combines cooperative UAVr with the LEO satellite to overcome NLoS challenges in densely populated urban settings. The framework employs amplify and forward techniques with cooperative diversity combining, allowing weighted SINR fusion of satellite–user and UAVr–user links to improve signal quality and network robustness. This dynamic approach addresses coverage gaps and congestion issues, significantly improving network capacity and energy efficiency compared to conventional systems lacking fairness or visibility awareness.

    \item An optimization problem is formulated to maximize the network's total energy efficiency and capacity while ensuring fairness among ground users. The proposed framework optimizes user association and transmission power allocation strategies, unifying these objectives within a single formulation. This integrated optimization ensures efficient resource utilization and improved performance in SAGINs.
    
\end{itemize}

%
\subsection{Organization}
Section~\ref{sec: Related work_3} presents and positions the work in the current literature. Section~\ref{sec: System Model_3} presents the system model, while Section~\ref{Adaptive UAVr Placement and Signal Combining} details the proposed AMUD framework. Section~\ref{sec: Problem Formulation_3} formulates a network energy efficiency maximization problem, and Section~\ref{The Proposed AMUD Framework} presents an algorithmic solution framework. Section~\ref{sec: Simulation Results and Performance Analysis_3} shows the simulation and performance analysis results, while Section~\ref{Sec: Conclusion_3} concludes the paper.

%
%
\begin{table*}[!ht]
\centering
\caption{Comparison of Representative Works on Relay-Based Access Architectures in SAGIN}
\label{Comparisons of Related Works and the Proposed Method_3}
\footnotesize 
\setlength{\tabcolsep}{6pt} 
\renewcommand{\arraystretch}{1.1} 
\begin{tabular}{c c c c c c c c}
\hline
\textbf{Ref.} &
\textbf{Relay Type} &
\textbf{Relay Count} &
\textbf{Trajectory Pattern} &
\textbf{Adaptive UAVr} &
\begin{tabular}[c]{@{}c@{}}\textbf{Satellite} \\ \textbf{Consideration}\end{tabular} &
\begin{tabular}[c]{@{}c@{}}\textbf{Satellite} \\ \textbf{Visibility}\end{tabular} &
\textbf{Capacity Fairness} \\ \hline
\cite{8949359,8894851} & Static & \begin{tabular}[c]{@{}c@{}}Single \\ Terrestrial\end{tabular} & -- & -- & Yes & -- & -- \\ \hline
\cite{zhao2018exact,6449258} & Static & \begin{tabular}[c]{@{}c@{}}Multi \\ Terrestrial\end{tabular} & -- & -- & Yes & -- & -- \\ \hline
\cite{9672696} & Static & \begin{tabular}[c]{@{}c@{}}Multi \\ Terrestrial\end{tabular} & -- & -- & Yes & -- & -- \\ \hline
\cite{8387982} & Mobile & Single UAVr & -- & -- & Yes & -- & -- \\ \hline
\cite{9714005,7899525} & Mobile & Single UAVr & Cyclical & -- & Yes & -- & -- \\ \hline
\cite{9062335} & Mobile & Single UAVr & Fixed & -- & Yes & -- & Yes \\ \hline
\cite{9755995} & Mobile & Multi UAVr & Cyclical & -- & Yes & -- & -- \\ \hline
\cite{9206550} & Mobile & Single UAVr & Fixed & -- & -- & -- & -- \\ \hline
\cite{9149409} & Mobile & Multi UAVr & Fixed & -- & -- & -- & -- \\ \hline
\textbf{Our Method} & \textbf{Mobile} & \textbf{Multi UAVr} & \textbf{Fixed} & \textbf{Yes} & \textbf{Yes} & \textbf{Yes} & \textbf{Yes} \\ \hline
\end{tabular}

\vspace{1mm}
\begin{minipage}{0.97\textwidth}
\footnotesize
\textit{Note:} This table focuses on relay-based architectures involving UAVr access link studies. As described in Section~\ref{sec: Related work_3}, UAVr nodes do not serve as direct access points but instead relay traffic via the nearest GBS or satellite.
\end{minipage}
\end{table*}


%
\section{Related Work}
\label{sec: Related work_3}
This section reviews recent literature on enhancing network capacity in SAGINs and positions the paper's contribution.
Table~\ref{Comparisons of Related Works and the Proposed Method_3} summarizes the related work and positions the proposed AMUD framework in terms of the type of relay, number of relays used, adaptive deployment of UAVr, access link, satellite visibility, and signal combination strategy employed.
Several recent studies have presented deployment methods to enhance network performance, considering static and mobile relay deployment without accounting for satellite visibility and adaptive relay deployment. 
%
%

Recent studies have focused on the deployment of static terrestrial relays to enhance network performance in designated regions by integrating them with satellite access links \cite{8949359,6449258,8894851,zhao2018exact,9672696}.
In \cite{8949359}, the authors examined the ergodic capacity of hybrid satellite-terrestrial networks (HSTNs) that incorporate multiple terrestrial relays in addition to a direct satellite link. They utilized asymptotic analysis for both AF and decode-and-forward (DF) protocols, employing moment-generating function (MGF) techniques. A relay selection strategy was proposed based on statistical information from channel state information (CSI) to reduce overhead.
The framework presented in \cite{6449258} employs a cooperative approach that integrates direct and terrestrial relay links to enhance diversity. Huang et al. \cite{8894851} formulated an optimization problem aimed at maximizing the system's capacity within HSTNs subject to shadowed Rician fading, deriving expressions for outage probability and ergodic capacity while introducing strategies for multi-user scheduling.
Zhao et al. \cite{zhao2018exact} investigated the ergodic capacity of a generalized HSTN featuring multiple relays and a direct satellite link. They applied MGF techniques for the AF and DF protocols and proposed a relay selection strategy that minimizes reliance on statistical CSI.
Finally, the work in \cite{9672696} explores the synergy between satellite and terrestrial relays in a dual-hop system, focusing on ergodic capacity and bit-error rate under varying channel conditions. An opportunistic relay approach was used to support multiple users effectively.
%
%
Recent studies have concentrated on strategies for the deployment of UAVs to improve network performance by using satellite access \cite{9755995,9062335,8387982,9714005,7899525,9206550,9149409}.
In \cite{9755995}, machine learning techniques were used to address link selection and UAV trajectory challenges within UAV-assisted hybrid satellite-terrestrial networks. Although notable improvements were achieved, the strategic placement of UAVs to alleviate congestion in satellite backhaul remains an open issue. Furthermore, incorporating cyclical patterns into UAV trajectory planning could optimize resource utilization and improve overall network efficiency.
The research presented in \cite{9062335} investigated multi-user asymmetric free-space optical (FSO) and radio frequency (RF) links between satellites and UAVs, accounting for Gamma-Gamma turbulence effects. The findings revealed a trade-off between fairness and system capacity when UAV trajectories were fixed.
In the realm of relay systems enabled by ultra-reliable low-latency communication (URLLC) assisted by UAVs, \cite{9206550} used the Nelder-Mead simplex search algorithm to optimize the height and positioning of UAVs in conjunction with the allocation of block length. This approach aimed to minimize decoding errors while maintaining fixed UAV paths in conjunction with reconfigurable intelligent surfaces (RIS).
Finally, \cite{9149409} introduced a Gibbs sampling method for optimal placement of UAVs as relay nodes, focusing on maximizing the minimum achievable rate for multiple pairs of ground nodes while keeping the UAV trajectories constant. 
%
%

To improve user QoS at the receiver end, recent research efforts (e.g.~\cite{8653373,6676778}) have proposed integrated techniques aimed at improving network performance, particularly through cooperative communication schemes.
The study in \cite{8653373} evaluates the effectiveness of cooperative communication using receiver-end selection combining (SC) on a Nakagami fading channel.
In \cite{6676778}, a novel cooperative space-shift keying network is introduced, featuring multi-antenna configurations for the source, destination, and relays. This approach involves selecting relays based on predetermined thresholds during the second time slot for data transmission from the source to the destination, where successive interference cancellation is employed for detection. The authors performed a comprehensive theoretical analysis of error performance, specifically addressing the scenario with two transmit antennas at the source and providing an analytical approximation for more general cases. 

As shown in Table~\ref{Comparisons of Related Works and the Proposed Method_3}, existing work has focused on optimizing traffic offloading to maximize user service and system capacity. These works primarily focus on single or multiple UAV-BS deployments with or without satellite access. Recent studies have yet to comprehensively address the effects of adaptive multiple UAVr deployment, satellite visibility, and signal-combining techniques.
Our proposed AMUD approach addresses all of these limitations. It jointly considers the adaptive deployment of multiple UAVr considering the visibility of the LEO satellite. The proposed framework also involves power control transmission on the UAVr and a combination of spacetime signals on the receiver using CD. We present the system model in the next section. 
%
%
\section{System Model}
\label{sec: System Model_3}
    The paper considers downlink resource allocation in a terrestrial network assisted by multiple UAVr and LEO satellites, as illustrated in Fig.~\ref{fig: system model_3}. The set of ground users is represented and modeled as a Poisson point process (PPP) within the GBS coverage area.
    In our simulation, urban hotspots are generated based on two criteria: (1) density-based, where the user density within a UAVr’s coverage area exceeds a predefined threshold \( \omega_j^{\max} = D_j^{\text{th}} \cdot \pi \cdot R_j^2 \), and (2) user-based, where the number of users in a GBS area exceeds its capacity, forming excess users \( \Omega_{\mathbb{G}}^{\text{Ex}} = \Omega_{\mathbb{G}} - \omega_{\mathbb{G}}^{\max} \). These dynamic conditions trigger UAV assistance for load balancing.
    We examine the dynamic behavior of the network on a set $\chi$ comprising $\varkappa$ time intervals indexed by $t$. The network configuration is considered stable given the brief duration of each time slot. We denote the two-dimensional coordinates of ground users in the time slot $t$ as $\mathbf{u}_i(t) = (x_i(t), y_i(t))$, where $i \in \{1, \ldots, N\}$ and $N$ indicates the total number of users in the system.
    The framework enables the deployment of up to $K$ multiple UAVr, represented as $\mathbf{U}_j(t) = (x_j(t), y_j(t), h^{\text{Ur}}_j(t))$, where $j \in \{1, \ldots, K\}$, within the coverage area. 
    ``It is important to note that although the user locations are monitored in real time, the deployment of UAVr is based on the final snapshot sampled within each interval. Once deployed, UAVr positions remain fixed, and thus the system behaves equivalently to a static deployment.''
    Let $p_{i  \leftarrow j}$ denote the minimum transmit power required to send a signal from the UAVr to the ground user.

    We ignore the effects of changes in the UAVr's three-dimensional ($3$D) position on the LEO satellite. The central control station (CCS) continuously monitors UAVr movement, determines user associations with multiple UAVr, and prevents collisions between them.
    Mobile users are considered hybrid \cite{9385374}, capable of communicating with LEO satellites, multiple UAVr, and GBS. Users communicate with multiple UAVr and LEO satellites on the same frequency band while using different frequencies for communication with GBS.
    Each UAVr is assumed to operate for a maximum of $30$ minutes per flight, reflecting practical battery constraints in real-world deployments~\cite{9555387}.
    We will discuss the signal and channel models in the upcoming subsections.
    
{
\noindent\textbf{Notational Convention:}
Scalars are written in regular font (e.g., $p_{i\leftarrow j}$, $\gamma_{i\leftarrow j}$), whereas all vectors and matrices are written in boldface (e.g., $\mathbf{h}_{i\leftarrow s}$, $\mathbf{n}_{i\leftarrow s}$, $\mathbf{u}_i$).
{Since each UAVr employs a single antenna, the satellite-to-UAVr channel $\hbar_{j\leftarrow s}$ and the corresponding received signal ${r}^{j \leftarrow s}$ are scalar quantities.}
}
%
\subsection{Signal model}
This subsection presents the signal model for the LEO-assisted multiple UAVr-based communication framework. The proposed framework pertains to a dual-hop cooperative diversity system comprising an LEO satellite ($s$), multiple UAVr ($\mathbf{U}_j$), and ground users ($\mathbf{u}_i$), each
equipped with $L$ antennas, as illustrated in Fig.~\ref{fig: system model_3}.
The channel coefficients are defined as follows: $\mathbf{h}_{i \leftarrow s}$ represents the complex channel coefficients from the satellite to the $l$th antenna of the ground user, $\hbar_{j\leftarrow s}$ denotes the coefficients between the satellite and the UAVr, and $\mathbf{h}_{i \leftarrow j \leftarrow s}$ corresponds to those between the UAVr and the ground user antennas.
Assuming that the satellite transmits a signal $x_{\rm sym}$ with an average power of $P_s^{\rm tx}$ during the first phase, the signal received in the UAVr from the satellite is given by
\[ {r}^{j \leftarrow s} = \hbar_{j\leftarrow s} x_{\rm sym} + n_{j\leftarrow s}.\]

{The system employs a TDMA-based scheduling strategy in which each time slot serves exactly one scheduled ground user. Although the transmit power is indexed by user and UAVr indices for notational generality, at any given time slot, only the scheduled user--UAVr pair is active. Consequently, only one transmit power variable is effective per slot, while all remaining power variables are inactive during that slot.
In Phase~I, the satellite transmits a single codeword intended for the scheduled user, which is received simultaneously by that user and by the associated UAVr. No intra-slot multi-user transmission occurs on either the satellite or UAVr links, and no bandwidth or power sharing among multiple users takes place within a time slot.}
In contrast, the signal received directly from the satellite to the user is given by
\[ \mathbf{r}^{i \leftarrow s} = \mathbf{h}_{i \leftarrow s} x_{\rm sym}
+ \mathbf{n}_{i \leftarrow s}. \]

In the second phase, the satellite remains silent, and the UAVr retransmits a scaled version of the received signal from the satellite in fixed-gain AF mode, with the received signal at the user expressed as

\begin{align*}
    \mathbf{r}^{i \leftarrow j \leftarrow s}
    &= \mathbf{h}_{i \leftarrow j \leftarrow s}\,
    \mathcal{G}\big({r}^{j \leftarrow s}\big)
    + \mathbf{n}_{i \leftarrow j \leftarrow s} \\
    &= \mathbf{h}_{i \leftarrow j \leftarrow s}\,
    \mathcal{G}\hbar_{j\leftarrow s} x_{\rm sym}
    + \mathbf{h}_{i \leftarrow j \leftarrow s}\,
    \mathcal{G}n_{j \leftarrow s}
    + \mathbf{n}_{i \leftarrow j \leftarrow s}.
\end{align*}

The above dual-hop cooperative communication framework at the user end is written as
\begin{align}
    \label{Eq: receive signal}
    \mathbf{r}^{\text{Tot}} = \mathbf{h} x_{\rm sym} + \mathbf{n},
\end{align}
where

$\mathbf{r}^{\text{Tot}} =
\begin{bmatrix}
\mathbf{r}^{i \leftarrow s} \\
\mathbf{r}^{i \leftarrow j \leftarrow s}
\end{bmatrix},
\quad
\mathbf{h} =
\begin{bmatrix}
\mathbf{h}_{i \leftarrow s} \\
\mathbf{h}_{i \leftarrow j \leftarrow s}\,
\mathcal{G}\hbar_{j\leftarrow s}
\end{bmatrix},
\quad
\mathbf{n} =
\begin{bmatrix}
\mathbf{n}_{i \leftarrow s} \\
\mathbf{h}_{i \leftarrow j \leftarrow s}\,
\mathcal{G}{n}_{j \leftarrow s}
+ \mathbf{n}_{i \leftarrow j \leftarrow s}
\end{bmatrix}$.
The noise terms $\mathbf{n}_{i \leftarrow s}$, $n_{j \leftarrow s}$, and
$\mathbf{n}_{i \leftarrow j \leftarrow s}$ denotes the AWGN at the ground user in the first hop, the scalar AWGN at the single-antenna UAVr, and the AWGN at the ground user in the second hop, respectively, and are modeled as $\mathcal{CN}(0,\sigma^2)$.
The fixed AF relay gain on the UAVr is defined as $ \mathcal{G}^2 = \frac{1}{|\hbar_{j\leftarrow s}|^2 + \sigma^2},$
where $\hbar_{j\leftarrow s}$ denotes the satellite-to-UAVr channel coefficient and $\sigma^2$ is the AWGN variance.

The AF relay gain $\mathcal{G}(t)$ models the fixed base band amplification applied to the received satellite signal at the UAVr during time slot $t$. The resulting AF output is normalized and subsequently scaled by a controllable RF transmit-power block. Therefore, the UAVr transmit power $p_{i\leftarrow j}(t)$ is treated as an independent optimization variable applied after AF processing and is not implicitly determined by $\mathcal{G}(t)$. 
Although $\hbar_{j\leftarrow s}(t)$ is a random variable, its realization is assumed to remain fixed within each time slot under the quasi-static fading model~\cite{6280241}. Therefore, the AF relay gain $\mathcal{G}(t)$ is fixed per slot and adapts only between slots as the channel evolves. This modeling ensures that the gain remains constant during AF operation within a slot while allowing time-dependent power optimization through $p_{i\leftarrow j}(t)$ across different time slots, consistent with the energy-efficiency objective.
Although the UAVr operates as a fixed-gain AF relay at the base band processing stage, the RF transmit power $p_{i\leftarrow j}(t)$ is treated as an independent optimization variable and is optimized on a per-slot basis subject to a maximum power constraint, as detailed in Section~\ref{sec: Problem Formulation_3}.

 \subsection{Channel Model}  
\label{sec:system_3}
    This subsection discusses the channel between the LEO satellite and the user, the LEO satellite to the UAVr, and the UAVr to the user. 
%
%
\subsubsection{UAVr to User}
The horizontal distance between the UAVr and the location of the user on the ground can be defined as (Fig.~\ref{fig: system model_3})
\begin{align}
    r_{i  \leftarrow j}(t) = \sqrt{(x_{j}(t) - x_{i})^2 + (y_{j}(t) - y_{i})^2}.
    \label{Eq:07_3}
\end{align}
Based on equation~\eqref{Eq:07_3}, the Euclidean distance between the UAVr and the ground user can be defined as
\begin{align}
    d_{i  \leftarrow j}(t) = \sqrt{r^2_{i  \leftarrow j}(t) + (h^{\rm Ur}_j(t))^2}.
    \label{Eq:08_3}  
\end{align}
The path loss to the user is computed using the air-to-ground channel model from \cite{al2014optimal}, based on line-of-sight (LoS) and NLoS conditions, given by
\begin{align}
    PL^\text{LoS}_{d_{i  \leftarrow j}}(t) &= 20\log_{10}\left(\frac{4\pi f_c d_{i  \leftarrow j}(t)}{c}\right) + \eta_{{\rm{LoS}}}, \nonumber\\
    PL^\text{NLoS}_{d_{i  \leftarrow j}}(t) &= 20\log_{10}\left(\frac{4\pi f_c d_{i  \leftarrow j}(t)}{c}\right) + \eta_{{\rm{NLoS}}}. \nonumber
\end{align}
Here, $\eta_{{\rm{LoS}}}$ and $\eta_{{\rm{NLoS}}}$ are the additional mean losses due to LoS and NLoS communication links \cite{7037248}; $c$ is the speed of light in meters per second; and $f_c$ is the carrier frequency in hertz. The corresponding probability of LoS signals from the UAVr to the ground user is given by
\begin{align}
    P^\text{LoS}_{d_{i  \leftarrow j}}(t) = \frac{1}{1 + a \exp\left(-b\left(\frac{180}{{\rm \pi }}\theta_{i  \leftarrow j} - a\right)\right)},  \nonumber
\end{align}
where $\theta_{i  \leftarrow j} = \tan^{-1}\left(\frac{h^{\rm Ur}_j(t)}{r_{i  \leftarrow j}(t)}\right)$ (in radians) is the elevation angle between the UAVr and the user, as shown in Fig.~\ref{fig: system model_3}; $a$ and $b$ are constant factors that depend on different environmental conditions (rural, urban, dense urban, etc.). For $P^\text{LoS}_{d_{i  \leftarrow j}}(t)$, the probability of NLoS signals from the UAVr to the ground user is given by $P^\text{NLoS}_{d_{i  \leftarrow j}}(t) = 1 - P^\text{LoS}_{d_{i  \leftarrow j}}(t)$.
The channel gain {$\mathbf{h}_{i \leftarrow j \leftarrow s}$} between the UAVr and the user is defined as \cite{10164260}
\begin{align}
    \label{eq: gain_UAV_user}
    \begin{split}
        {\mathbf{h}_{i\leftarrow j \leftarrow s}}(t) = & {\mathfrak{g}_{i \leftarrow j}}{{\left( \frac{4\pi {{f}_{c}}{{d}_{i\leftarrow j}(t)}}{c} \right)}^{-\frac{{{\alpha_{\rm exp} }}}{2}}} \\
        & \times {{10}^{-{\frac{P_{d_{i \leftarrow j}}^{{\rm{LoS}}}(t)\times PL_{d_{i \leftarrow j}}^{{\rm {LoS}}}(t)+P_{d_{i \leftarrow j}}^{ {\rm{NLoS}}}(t)\times PL_{d_{i \leftarrow j}}^{{\rm{NLoS}}}(t)}{20}}}}.
    \end{split}
\end{align}
where $d_{i \leftarrow j}$ is the distance between the UAVr and the user, $\mathfrak{g}_{i \leftarrow j}$ denotes the channel gain that includes both small-scale fading and the effect of antenna gain, following the stochastic modeling convention,~{and $\alpha_{\rm exp}$ denotes the path-loss exponent of the UAVr to user link.}
The average path loss of the signal from the UAVr to the ground user is computed as 
    \begin{align}
        &PL_{d_{i \leftarrow j}}^{{\rm{Avg}}}(t) = P_{d_{i \leftarrow j}}^{{\rm{LoS}}}(t)\times PL_{d_{i \leftarrow j}}^{{\rm {LoS}}}(t) + P_{d_{i \leftarrow j}}^{ {\rm{NLoS}}}(t)\times PL_{d_{i \leftarrow j}}^{{\rm{NLoS}}}(t) \notag \\
        & = \frac{\eta_{\rm LoS} - \eta_{\rm NLoS}}{{1+a\exp\left({-b\left[{{\frac{180}{{\rm \pi }}\theta_{i  \leftarrow j}}-a}\right]}\right)}} + 20{\log_{10}}\left ({{d_{i  \leftarrow j}(t)}}\right) + \beta,
    \label{Eq:09_3} 
    \end{align}
where $\beta = 20{\log_{10}}\left(\frac{4\pi f_c}{c}\right) + \eta_{{\rm NLoS}}$.
We also neglect satellite interference due to its relatively low power compared to the desired signal strength for advanced IoT devices~\cite{8489991}. 
%
Thus, the SINR for the user associated with UAVr can be defined as
\begin{align}
    \label{Eq:10_3}
    \gamma_{i \leftarrow j}(t/2) = \frac{p_{i \leftarrow j}(t)\,{\left\lVert \mathbf{h}_{i \leftarrow j \leftarrow s}(t) \right\rVert}^{2}} {B_{i \leftarrow j}\,\sigma^{2} + I_u(t)} .
\end{align}

where $I_u(t) = \sum_{\substack{k=1, k \neq j}}^{K} p_{i\leftarrow k}(t)\, 10^{PL^{\rm Avg}_{d_{i \leftarrow k}}(t)/10}$ is the aggregate interference from all UAVr $k\neq j$ on the user $i$, and $B_{i \leftarrow j}$ denotes the bandwidth allocated to the UAVr$\to$user link.
\noindent Here, $(t/2)$ marks the phase (I/II) within the slot; all instantaneous quantities are evaluated at time $t$. 
Accounting for the two-phase schedule is handled later in the rate expression.
To guarantee the QoS requirement, we impose {$\gamma_{i\leftarrow j} \ge \gamma_{\rm th}$,} where ${\gamma }_{\mathrm{th}}$ is the predefined SINR threshold for successful transmission of the signal. 
The hovering power of UAVr as a function of its operational altitude \( h^{\rm Ur}_j \) is calculated as 
\begin{align}
    p_j^{\rm Hov} = p_0(1+\Delta)e^{{{\varepsilon h^{\rm Ur}_j}/{2}}},
    \label{Eq:01_3}
\end{align}
where \( p_0 \) is the power consumed by the serving UAVr during hovering; \( \Delta \) and \( \varepsilon \) are constants; and \( h^{\rm{Ur}}_j \) is the altitude of the UAVr.
The exponential form adopted from~\cite{9946428} originates from an aerodynamic drag model, which inherently reflects altitude-related resistance and static environmental factors such as air density.
The hovering altitude of the UAVr corresponding to its hovering power, derived from~(\ref{Eq:01_3}), is given by
\begin{align}
    h^{\rm Ur}_j = \frac{2}{\varepsilon}{\ln}\left(\frac{{p_j^{\rm Hov}} }{{p_0(1+\Delta)}}\right).
    \label{Eq:02_3}
\end{align}

Since the altitude of UAVr is fixed at a regulatory-compliant 100 \,m in all schemes, the hover power acts as a constant offset in total energy consumption. Therefore, this constant does not alter the results of relative performance or the conclusions drawn from our comparative analysis.

    \begin{figure}[!ht]
	   \centering
	   \includegraphics[width=0.99\columnwidth]{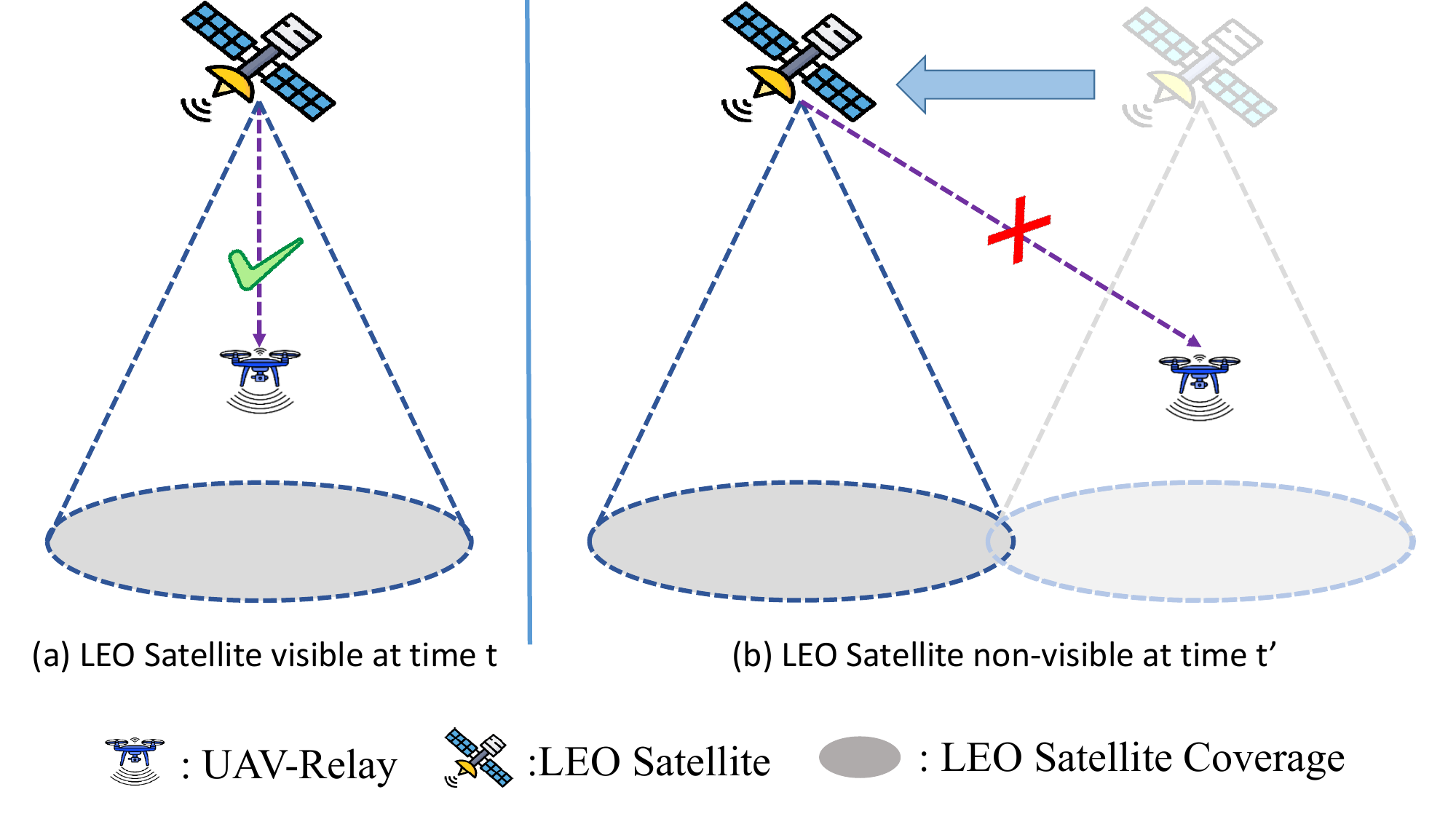}
       \caption{\small Illustration of LEO satellite visibility states. LEO satellites travel at high orbital speeds (typically 7.8\,km/s~\cite{9681631}), causing rapid changes in visibility to ground or aerial nodes.}
	   \label{fig: Satellite Visibility}
    \end{figure}
    
    %
   \subsubsection{ LEO Satellite to UAVr and User }  

  Due to the high-speed orbits of LEO satellites and their limited visibility periods to multiple UAVr, a satellite constellation is essential to ensure continuous backhauling capabilities. Orbital pass periods during which LEO satellites are directly visible to multiple UAVr are a critical factor influencing backhaul efficiency \cite{zong2019optimal}. Since LEO satellites are not geostationary, their movements must be analyzed before evaluating the system's performance. We assume that each UAVr is continuously associated with an LEO satellite. For computational convenience regarding satellite-Earth geometries, we model the satellite following a circular orbit. 
  The connection between a UAVr and an LEO satellite is represented by a visibility variable \( V_{j  \leftarrow s}(t) \) at the time interval \( t \), as defined in \cite{9775682}.
    \begin{equation}
     \label{eq: Satellite_visibility}
        {V }_{j\leftarrow s}(t)=\left\{ \begin{array}{c} 1\ \mathrm{\ }\mathrm{\ }\mathrm{if}\ {\mathrm{cos} \left(\frac{\mathrm{2}\mathrm{\piup }\mathrm{t}}{T_{\mathrm{s}}}-{\theta }_p\right)\ge \ }\frac{R^2_E+{r_{\rm EC}^2}-d^2_{\mathrm{SR}}}{2R_E{r_{\rm EC}}} \\ \ 0\ \ \ \ \ \ \ \ \ \ \ \ \ \ \ \ \ \ \mathrm{otherwise}.  \ \ \ \ \ \ \ \ \ \ \ \end{array} \right.
    \end{equation}  
    {To avoid potential ambiguity, satellite visibility is evaluated at the beginning of each time slot and assumed constant over the slot duration.}
    If the LEO satellite is visible at time \( t \), then \( V_{j  \leftarrow s}(t) = 1 \); otherwise, \( V_{j  \leftarrow s}(t) = 0 \) \cite{10464433}, as shown in Fig.~\ref{fig: Satellite Visibility}.  Similarly, the visibility variable from the user to the LEO satellite, \( V_{i \leftarrow s}(t) \), is determined in the same manner as in equation~\eqref{eq: Satellite_visibility}.
    Here, \( d_{\rm SR} \geq d_{j  \leftarrow s} \) represents the slant range corresponding to the minimum elevation angle, \( R_E \) is the radius of the Earth, \( r_{\rm EC} \) is the distance between the LEO satellite and the Earth's center, \( \theta_p \) is the polar angle of the UAVr, and \( T_s \) is the orbital period of the satellite. 
    It is assumed that the handover between LEO satellites and multiple UAVr occurs through advanced handover (HO) schemes, such as guaranteed and prioritized HO; therefore, frequent handovers due to LEO satellite movement do not affect data transmission reliability \cite{4062836}.
    We assume that the links between the LEO satellite and the UAVr or the user follow the shadowed-Rician fading (SRF) model.
   The channel gain from the satellite to the UAVr is expressed as \cite{10164260}:
    \begin{align}
        \hbar_{j\leftarrow s} = \sqrt{\mathfrak{g}^{\rm avg} d_{j\leftarrow s}^{-\alpha_{\rm exp}}},
    \end{align}
    where the term $\mathfrak{g}^{\text{avg}}$ represents the average antenna gain between the LEO satellite and the UAVr/user and captures the overall gain and fading behavior in the satellite communication links.
    \( d_{j  \leftarrow s} \) is the distance between the satellite and the \( j^{\text{th}} \) UAVr, and \( \alpha_{\rm exp}^2 \) is the path loss exponent from the satellite to the UAVr. 
    The term \( \mathfrak{g}^{\rm {avg}}\sim\text{SR}(\wp, \Im, \varnothing) \) represents the SRF component with an average power of the direct signal \( \wp \), half of the average power of the scattered portion \( \Im \), and the Nakagami-m fading component \( \varnothing \).

    The instantaneous SNR for the satellite-to-UAVr link, based on the satellite visibility criterion in~\eqref{eq: Satellite_visibility}, is given as
    \begin{equation}
       \label{eq: Satellite_UAV}
       \gamma_{j \leftarrow s}(t/2) = 
       ({V_{j \leftarrow s}(t)\, P_s^{\rm tx}\, 
             \lvert \hbar_{j\leftarrow s}(t) \rvert^{2}})/({B_{j \leftarrow s}\, \sigma^2} \;) .
    \end{equation}
    \noindent The same convention as in Eq.~(\ref{Eq:10_3}): $(t/2)$ is a phase indicator; instantaneous terms are taken at time $t$. 
    %
    Similarly, the channel coefficient for the satellite-to-user link is modeled as
\[
\mathbf{h}_{i \leftarrow s}(t) = \sqrt{{\mathfrak{g}}^{\rm avg}\, d_{i\leftarrow s}(t)^{-\alpha_{\rm exp}}},
\]
and the corresponding SINR is given by
\begin{equation}
    \label{eq: Satellite_User}
    \gamma_{i \leftarrow s}(t/2)
    = V_{i\leftarrow s}(t)\,
    \frac{P_s^{\rm tx}\, \bigl\lVert \mathbf{h}_{i \leftarrow s}(t) \bigr\rVert^{2}}
         {B_{i \leftarrow s}\, \sigma^2 + I_u(t)} \, .
\end{equation}
Here, $P_s^{\rm tx}$ is the satellite transmission power, $B_{i\leftarrow s}$ denotes the bandwidth allocated to the satellite-to-user link during its active phase, and $(t/2)$ marks the phase (I/II) within the slot; instantaneous quantities are evaluated at time $t$. 
    It is noted that the satellite–UAVr and satellite–user distances vary with orbital motion across time intervals, and this variation is incorporated into the SINR calculation through \eqref{eq: Satellite_visibility}, \eqref{eq: Satellite_UAV}, and \eqref{eq: Satellite_User}. The UAVr positions are assumed to be constant within a time slot for tractability.


    
\subsubsection{ GBS to User }  
   The user is assumed to experience independent Rayleigh fading when associated with the GBS. In each time $t$, the channel coefficient is: 
   \begin{align}
       \mathbf{h}_{i \leftarrow G}\left(t\right)={ \mathfrak{g}_{i \leftarrow G}\left(r_{i \leftarrow G}\left(t\right)\right)}^{-\alpha_{\rm exp} }.
       \label{Eq:20_3}
    \end{align}
    Here, \( \mathfrak{g}_{i \leftarrow G} \) denotes the channel gain that captures both the antenna gain and small-scale Rayleigh fading, modeled as a complex Gaussian random variable, i.e., \( \mathfrak{g}_{i \leftarrow G} \sim \mathcal{CN}(0,1) \).
    The variable \( r_{i \leftarrow G} \) represents the horizontal distance between the user and the GBS. Thus, the SNR for the user associated with the GBS can be defined as follows:
    \begin{align}
        {\gamma }_{i \leftarrow G}\left(t\right)=\frac{P_{\!G}^{\rm{tr}}{\left\|\mathbf{h}_{i \leftarrow G}\left(t\right)\right\|}^2}{B_{i \leftarrow G}{\sigma }^2}.
        \label{Eq:21_3}
    \end{align}
    The GBS transmits at a fixed power \( P_{\!G}^{\,\mathrm{tr}} \)  for terrestrial communications. To maximize total network capacity while ensuring fairness, UAVr/satellite links are utilized only when the GBS cannot accommodate additional users. We define \( \delta_{i \leftarrow G}(t) \) as a binary indicator for the user at time \( t \), indicating whether the user meets the load conditions under GBS coverage:
\begin{equation}
    \delta_{i \leftarrow G} = 
    \begin{cases}
      1, & \text{if} ~~~~ |\Omega_{G}| \leq \omega_{G}^{\text{max}} \\
      0, & ~~~~~~~\text{otherwise.}
    \end{cases}
    \label{Eq:gbs_indicator}
\end{equation}
Here, $\Omega_{G}$ and $\omega_{G}^{\text{max}}$ show the current users associated with the GBS and the maximum user association capacity of the GBS.
   \begin{figure}[!ht]
       \vspace{-5pt}
	   \centering
	   \includegraphics[width=.96\columnwidth]{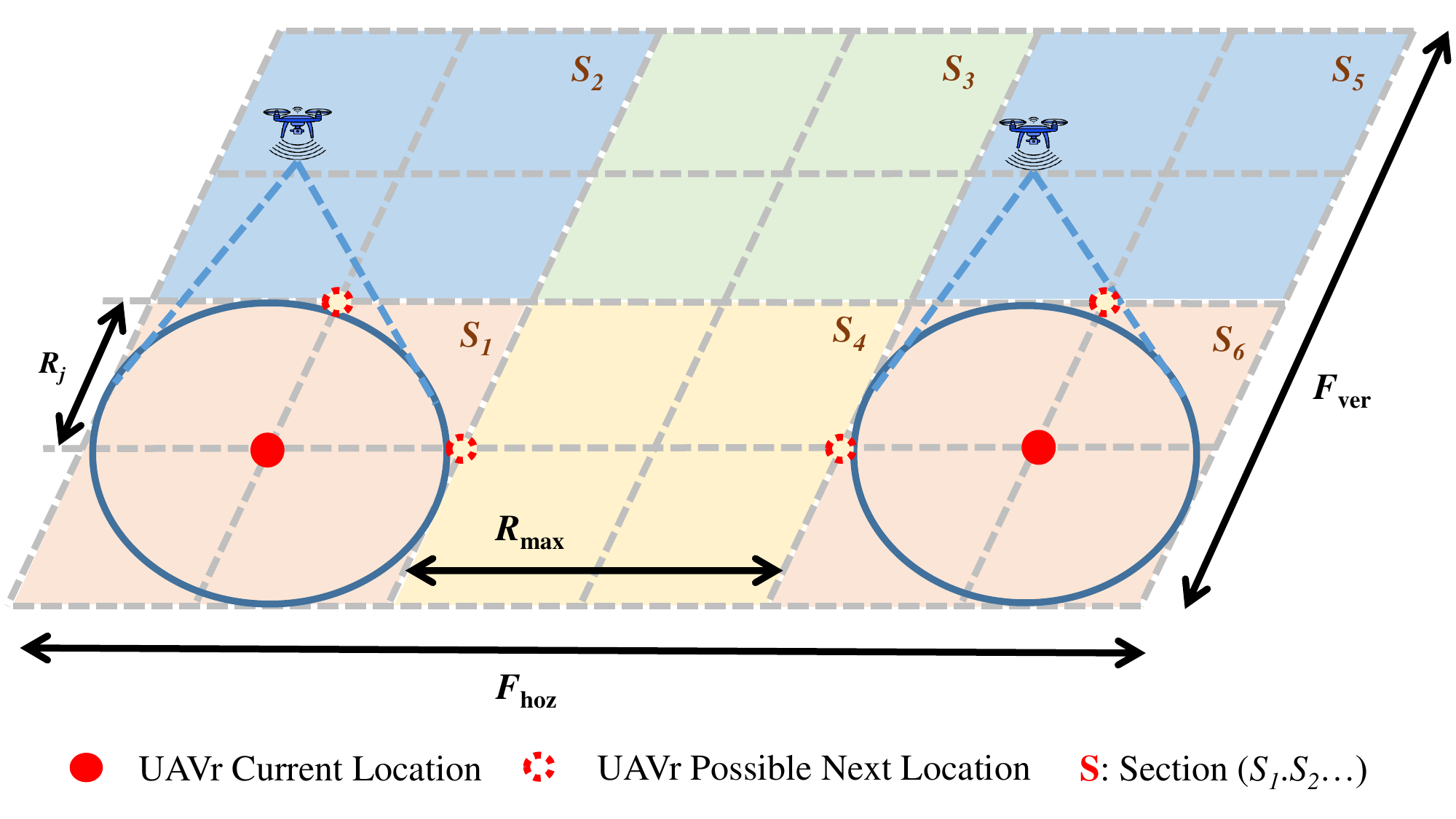}
	   \vspace{-5pt}
       \caption{\small Illustration of the UAVr mobility model. Each UAVr provides circular coverage with radius $R_j$ (blue disk), while the rectangular grid represents logical subdivisions of the GBS area for hotspot tracking and UAVr displacement planning. The grid does not imply rectangular coverage but is used for section-wise traffic analysis and decision-making. The UAVr selects its next location within its circular reach based on traffic demand.}
	   \label{fig: UAV-mobility}
    \end{figure}
    %
    

\section{ UAVr Deployment and Signal Combining}
    \label{Adaptive UAVr Placement and Signal Combining}
    When GBS cannot meet the QoS requirements of network users, the proposed AMUD approach first involves intelligently deploying multiple UAVr over the hotspot area and then collaboratively serving the users with the help of the LEO satellite. In this section, we first analyze the adaptive placement of multiple UAVr, followed by a signal-combining analysis of the LEO satellite and multiple UAVr.

\subsection{UAVr Coverage Analysis}
    \label{UAVr coverage analysis}
   In this subsection, we discuss the deployment of multiple UAVr through the proposed AMUD framework. As illustrated in Fig.~\ref{fig: UAV-mobility}, the coverage area of a UAVr is considered circular with a radius of \( R_j \). Let the terrestrial area under GBS coverage be denoted as \( A \), with dimensions \( F_{\rm hoz} \times F_{\rm ver} \). The area \( A \) is subdivided into smaller sections that represent the coverage of multiple UAVr, denoted \( S_A \), where \( S_A < A \). A key objective of the proposed AMUD framework is to deploy multiple UAVr in sectors experiencing heavy traffic. Furthermore, intelligent user association must be performed to ensure that each user connects to the most relevant access point. 
   A user is considered to be within coverage if it is located within a distance of at most \( R_j \) from the center of the coverage region, i.e.,
    \begin{align}
       \delta_{i \leftarrow j}r^2_{i  \leftarrow j}(t)\le R^2_j(t).
        \label{Eq:25_3}
    \end{align}
    We further modify the above equation as follows~\cite{7918510}, 
    \begin{align}
        r^2_{i  \leftarrow j}(t)\le R^2_j(t) + M(1-\delta_{i \leftarrow j}).
        \label{Eq:26_3}
    \end{align}
    Here, \( M \) is a large constant, indicating that the user is far from the access point when \( \delta_{i \leftarrow j} = 0 \). The equation above signifies user association when \( \delta_{i \leftarrow j} = 1 \). Thus, \( \delta_{i \leftarrow j} \in \{0,1\} \) serves as an indicator function that represents the association of a user with an access point. Mathematically, it can be expressed as follows:
    \begin{equation}
        \delta_{i \leftarrow j} = 
        \begin{cases}
            1, \ &\text{if,} \ r^2_{i  \leftarrow j}(t)\le R^2_j(t) \\
            0, \ &\text{otherwise.} 
        \end{cases}
    \label{eq: Indicater_UAV_user}
    \end{equation}
    %
    
    In a multiple-UAVr deployment framework, the multiple UAVr traverse the entire network region while hovering at the same height. Therefore, maintaining a safe distance is crucial to prevent collisions between multiple UAVr. A collision between two UAVr occurs when their distance is \( d_{j \leftarrow j^{\prime}} = 0 \)\cite{9555387}. To mitigate this risk, the proposed AMUD framework incorporates a safe distance parameter \( R_{\max} \) among multiple UAVr, ensuring that \( d_{j \leftarrow j^{\prime}} \geq R_{\max} \), where \( j \) and \( j^{\prime} \) denote different UAVr (Fig.~\ref{fig: UAV-mobility}). 

    Please note that for tractability, UAVr repositioning within a time slot is assumed instantaneous, representing the most responsive case; in practice, repositioning delays may slightly reduce achievable gains.
    The following subsection discusses the space-time CD signal combining strategy within the proposed AMUD framework.
    %

\subsection{Cooperative Diversity Technique }
   Following the adaptive deployment of UAVr discussed in the previous subsection, we would like to explore the combined CD-based signal at the ground user within the proposed AMUD framework. The user is assumed to possess \( L \) antennas.
   In this scenario, the UAVr is retransmitting the signal from the LEO satellite to the destination nodes, the ground user \cite{labrador2009approach}. 

    Communication from a user through the LEO satellite-aided UAVr occurs as follows. The proposed AMUD approach combines the signals received by the user in two phases. The LEO satellite and UAVr signals to the ground user arrive via time division multiple access (TDMA) protocols from Phase I to Phase II, as illustrated in Fig.~\ref{fig: system model_3}.
    During Phase I, the source (i.e., LEO satellite) simultaneously transmits a signal to both the relay node (i.e., UAVr) and the destination node (i.e., ground user). Subsequently, in Phase II, using the AF protocol, the UAVr retransmits the source signal to the user while the satellite remains silent.
    Additionally, we assume that no information exchange occurs between a UAVr and a LEO satellite between Phase I and Phase II while operating in time division duplex mode; therefore, the total time slot is two \cite{hong2010cooperative}. 
    Satellite and UAVr signals are assumed to be perfectly in synchronization with the user \cite{zhao2018exact}. This synchronization is achieved through timing calibration aided by terrestrial control stations and timing advance adjustments based on device locations \cite{9099807}.
    It should be noted that the AF protocol offers a more straightforward implementation and lower complexity than DF. Recognizing this advantage, our framework leverages AF for efficient signal relaying \cite{labrador2009approach}.
    AF relays are adopted here for their low latency and simple implementation compared to DF~\cite{6280241}, though they inherently forward both signal and noise. Our SINR-weighted combining \eqref{eq: Satellite_visibility}–\eqref{eq: Satellite_User} assigns lower weights to noisier relay paths, mitigating noise amplification, while the CD between the satellite and UAVr links further suppresses residual noise.

    %
    
    Let \( \mathbf{w} = \left( w_{1_{1}}, w_{1_{2}}, \ldots, w_{1_{L}}, w_{2_{1}}, w_{2_{2}}, \ldots, w_{2_{L}} \right) \in \mathbb{C}^{2L} \) be a weighting vector used to combine \( 2L \) received signals from the LEO satellite and the UAVr as defined in equation~\eqref{Eq: receive signal}. Then, using the CD receiver with the weighting vector \( \mathbf{w}^{\dag} \), we can express the combined output \( \mathbf{w}^{\dag} \mathbf{r_{\text{s}}^{\text{Tot}}} \) to the user as,
    %
    \begin{align}
    \label{eq: weighted_MRC}
        &\mathbf{w}^{\dagger} \mathbf{r_{\text{s}}^{\text{Tot}}}=
        \underbrace{\mathbf{w}^{\dagger} {\textbf{h}} \mathbf{x}_{\text{sym}}}_{\text{Signal}} + \underbrace{\mathbf{w}^{\dagger} {{\textbf{n}}}}_{\text{Noise}}.
    \end{align}
    %
    where the superscript $(\cdot)^{\dagger}$ indicate the transpose of vector \textbf{w}.
    In the presence of complete CSI at the destination, the instantaneous SINR at the user can be expressed as
    \begin{align}
        {\gamma_{\text{\rm AF}}}(\mathbf{w})={P_{s}^{\text{\rm tx}}} \frac{ \mathbf{w}^{\dagger} R_{\rm k} \mathbf{w}}{{{\mathbf{w}}^{\dagger }} {\mathscr{R}_{\text{n}}} \mathbf{w}} 
       \label{eq: SNR_combiner} 
    \end{align}
    {
    This expression models the optimal combination on the receiver side under a fixed AF relay gain $\mathcal{G}(t)$, and does not imply any adaptive gain control on the UAVr.
    }

    \begin{proof}
        Please refer to Appendix~\ref{appendix-A_3}.
    \end{proof}
    %
         Taking the derivative of~\eqref{eq: SNR_combiner} to the weight vector $\mathbf{w}$ ~\cite{holter2002optimal}, we get the optimal beamforming weight vector in a dual-hop cooperative communication system given by  
    \begin{align}
        \mathbf{w}_{\rm opt} = c_{r} {\textbf{h}}' = c_{r} \mathscr{R}_{n}^{-1} {\textbf{h}}.
    \label{eq: optimal weight_1}
    \end{align}
    %
    Hence, the maximum SINR in a dual-hop amplify and forward communication system is given as 
    \begin{equation}
         \gamma_{\text{\rm {AF, max}}}^{\text{\rm CD}}(\mathbf{w}_{\text{\rm opt}}) = \gamma_{i \leftarrow s} + \frac{\gamma_{j \leftarrow s}\gamma_{i \leftarrow j}}{\gamma_{i \leftarrow j} + \varsigma}.
    \label{eq: SNR_max_1}     
    \end{equation}
    %
    \begin{proof}
        Please refer to Appendix~\ref{appendix-B_3}.
    \end{proof}
    
    To indicate whether the user is associated with the \( \mathbf{U}_j \)-th UAVr, the indicator function \( \delta_{i \leftarrow j} \) is modified to incorporate the QoS and coverage constraints of the user. 
   \begin{equation}
        \delta_{i \leftarrow j} = 
        \begin{cases}
            1, \ &\text{if,} \ \left(\gamma_{\text{AF, max}}^{\text{CD}}\left(t\right)\ge {\gamma }_{\mathrm{th}}\right) \mathrm{\wedge } \ ( r^2_{i  \leftarrow j}(t)\le R^2_j(t)) \\
            0, \ &\text{otherwise.} 
        \end{cases}
        \label{Eq:27_3}
    \end{equation} 
     It is assumed that each user can only connect to one UAVr at a time for fair user association, and such a constraint is written as
    \begin{align}
	   \sum_{j=1}^{K}\delta_{i \leftarrow j}=1, 
	   \label{Eq:13_3} 
    \end{align}
        
    \subsection {Total Capacity Calculation}
    In the two-phase TDMA cooperative scheme illustrated in Fig.~\ref{fig: system model_3}, Phase~I corresponds to satellite transmission and Phase~II to UAVr forwarding, with each phase occupying half of the duration of the slot. The achievable data rate obtained from \eqref{eq: SNR_max_1} based on Shannon's theorem is,
    \begin{align}
        \label{Eq:14_3}
        c^{\mathrm{CD}}_{\mathrm{AF,max}}(t)
        = ({1}/{2})\, B_{i \leftarrow j}\,
          \log_2\!\bigl(1 + \gamma_{\text{AF,max}}^{\text{CD}}(t)\bigr)\,
          \delta_{i \leftarrow j}(t) \, .
    \end{align}
    \noindent The factor ${1}/{2}$ reflects the two equal-duration phases in the schedule.
    \noindent$B_{i \leftarrow j}$ denotes the bandwidth allocated to the active downlink (e.g., UAVr$\to$user) during its assigned phase; analogous definitions apply for other links when composing $\gamma_{\text{AF, max}}^{\text{CD}}(t)$.
    According to equations~\eqref{Eq:gbs_indicator} and~\eqref{Eq:14_3}, the data transmission rate achievable by users associated with LEO satellites through UAVr is
    \begin{equation}
        C_j(t) = \sum_{i \in \Omega_j, \forall i \in \{1,2,\dots,N_{\text{U}}\}} c^{ \mathrm{CD}}_{\mathrm{AF,max}}(t),
    \label{Eq:16_3}
    \end{equation}
    where \( \Omega_j = \max(0, |\Omega_{\rm Tot}| - \omega_{G}^{\text{max}}) \) is the set of users associated with the collaboration between the UAVr and the LEO satellite, \( \Omega_{\rm Tot} \) is the total set of users in the system, \( \omega_{G}^{\text{max}} \) is the maximum user association capacity of the GBS and \( N_{\rm U} \) is the number of users collaboratively between the UAVr and LEO satellite. The achievable data rate for a user associated with the GBS, as given by~\eqref{Eq:21_3}, is
    \begin{align}
        c_{i \leftarrow G}(t) = B_{i \leftarrow G} \log_2(1+\gamma_{i \leftarrow G}(t))\cdot \delta_{i \leftarrow G}(t),
        \label{Eq:22_3} 
    \end{align}
    where \(B_{i \leftarrow G}\) is the bandwidth (MHz) allocated to the GBS-associated user. The data transmission rate of GBS~\eqref{Eq:22_3} serving its associated users is calculated as follows.
    \begin{equation}
        C_G(t) = \sum_{i \in \Omega_G, \forall i \in \{1,2,\dots,N_{\text{G}}\}} c_{i \leftarrow G}(t),
        \label{Eq:23_3}
    \end{equation}
    where \(\Omega_{G}\) is the set of users associated with GBS and $N_{\rm G}$ is the number of users served by GBS. 
    From~\eqref{Eq:16_3} and~\eqref{Eq:23_3}, the total capacity of the network is given as
   \begin{align}
         C^{\text{Tot}}(t) = \sum_{{j}=1}^{K}~ C_{j}(t) + C_G(t).
        \label{Eq:24_3} 
    \end{align}

   The total power consumption of the downlink system, including the communication power of the UAVr, the hovering power of the UAVr, the transmission power of the satellite, and the transmission power of the GBS, is given by 
\begin{equation}
    p^{\text{Total}}(t) = \underbrace{\sum_{j=1}^{K}\!\left(\sum_{i=1}^{N_{U}} p_{i \leftarrow j}(t) + p_{j}^{\text{Hov}}(t)\right)}_{\text{\scriptsize UAVr Power (Comm.\,+\,Hover)}}
    + \underbrace{P_{s}^{\text{tx}}(t)}_{\text{\scriptsize LEO Power}}
    + \underbrace{\sum_{i=1}^{N_{G}} P_{G}^{\text{tr}}(t)}_{\text{\scriptsize GBS Power}}
\label{Eq:15_3}
\end{equation}

    To clarify the GBS power model, we assume a fixed transmit power spectral density (PSD), resulting in a constant transmit power \( P_{\!G}^{\,\mathrm{tr}} \) for each user's allocated bandwidth. Consequently, the total power consumption of the GBS, denoted as \( P_{\!G}^{\,\mathrm{tr}} \)(t), increases linearly with the number of users associated with the time interval $t$. This relationship is explicitly captured in \ref{Eq:15_3}, ensuring that the model accounts for both operational overhead and dynamic user load.
    From~\eqref{Eq:24_3} and~\eqref{Eq:15_3}, we can derive the total energy efficiency as
    \begin{align}
        E^{\rm Tot}(t)= \frac{C^{\text{Tot}}(t)}{p^{\rm Total}(t)} 
    \label{eq: energy_efficiency}
    \end{align}
    %
    
%
We introduce a capacity-fairness constraint to ensure user fairness in a multi-access-point framework. We use Jain's fairness index (JFI) \cite{jain1984quantitative} metrics to measure performance and provide fair communication services to different users in the deployed UAVr-aided hybrid space-terrestrial network. Fairness among users is demonstrated using JFI. This index, denoted by \( \xi \), serves as a fairness metric and is defined as
    \begin{align}
        \xi =\frac{{{\left( \sum\limits_{i=1}^{{{N}_{U}}}{c_{\text{AF,max}}^{\text{CD}}+\sum\limits_{i=1}^{{{N}_{G}}}{{c_{i \leftarrow G}}}} \right)}^{2}}}{N\left( \sum\limits_{i=1}^{{{N}_{U}}}{c{{_{\text{AF,max}}^{\text{CD}}}^{2}}} \right)+N\left( \sum\limits_{i=1}^{{{N}_{G}}}{{c_{i \leftarrow G}}^{2}} \right)}. 
	    \label{Eq:28_3} 
    \end{align} 
The index value \( \xi \) is in the range \([0, 1]\). A value of one indicates equal data rates among different network users. A higher value of the fairness index corresponds to smaller differences between the total data rates allocated to the users. In the next section, we formulate a network capacity maximization problem that incorporates the constraints introduced here.
%
%
\section{Energy Efficiency Maximization Formulation} 
\label{sec: Problem Formulation_3}
In this work, we improve the network's total energy efficiency using the proposed AMUD approach with adaptive UAVr placement and signal combining from LEO satellites and UAVr. Due to the limited battery capacity of the UAVr that affects hovering time and communication efficiency, energy efficiency is critical, especially in disaster and temporary communication scenarios~\cite{9946428}. Consequently, we jointly optimize user association and transmit power allocation. The formulation of the problem, subject to the minimum data rate constraints~\eqref{Eq:14_3} and~\eqref{Eq:22_3}, the association of users, and the considerations of fairness, is given in~\eqref{eq: energy_efficiency}.
\begin{equation}
    \max_{\delta_{i \leftarrow j} (t), p_{i  \leftarrow j} (t)} \sum\limits_{i=1}^{T} E^{\rm Tot}(t) \; 
    \label{Eq:obj_P1_3}
\end{equation}

The multi-objective problem is scalarized using an $\alpha$-weighted sum of the normalized total capacity and normalized total energy efficiency, with fairness and QoS constraints enforced as hard constraints.
\vspace{-6.5pt}
\begin{subequations}
    \renewcommand{\theequation}{\theparentequation\alph{equation}}
    \begin{align}
    \hspace{3em}\text{subject to}&~\eqref{Eq:gbs_indicator},~\eqref{eq: SNR_max_1},~\eqref{Eq:27_3},~\eqref{Eq:13_3},~\eqref{Eq:14_3},~\eqref{Eq:22_3},\nonumber\\
        0 & \leq p_{i  \leftarrow j} \leq p_{\max}, 
        \label{Eq:30_3} \\
        0 & \leq |\Omega_{j}| \leq \omega_{j}^{\mathrm{max}}, 
        \label{Eq:31_3} \\
        0 & \leq |\Omega_{G}| \leq \omega_{G}^{\mathrm{max}}, 
        \label{Eq:32_3} \\
        &~r^2_{i  \leftarrow j}(t)\le {R^2_{j}(t)}, \notag\\
        &~{\delta}_{i \leftarrow j}\in \{ 0,1 \} , \ \forall i\in {\Omega_j}, 
        \label{Eq:33_3} \\
        &~{\delta}_{i \leftarrow G}(t) + \delta_{i \leftarrow j}(t) = 1  \ \forall i,
        \label{Eq:new_constraint} \\
        &~\xi\ge\xi_{\rm th}. 
        \label{Eq:34_3} \\
        &~\mathbf{U}_j(t), \forall t \in \{1, T\}. 
        \label{Eq:35_3}
    \end{align}
\end{subequations}
    
The constraint ~\eqref{Eq:30_3} limits the transmit power at each UAVr serving associated users within its coverage, influenced by factors such as the altitude of the UAVr, the service duration,  and the number of users. Constraints~\eqref{Eq:31_3} and~\eqref{Eq:32_3} restrict user associations with the UAVr and GBS, respectively. The system identifies candidate configurations, including UAVr locations \( \mathbf{U}_j^{\rm 2D} \) and optimal coverage radius \( R_j \), based on relationships defined in equation~\eqref{Eq:33_3}, derived from equation~\eqref{Eq:27_3}. The constraint~\eqref{Eq:new_constraint} ensures that each user is associated with either the GBS or the UAVr. 
The constraint~\eqref{Eq:34_3} ensures a fair distribution of users among multiple UAVr and GBS, with the fairness constraint $\xi \ge \xi_{\mathrm{th}}$ directly applied to the feasible set, thereby maintaining fairness independently of the scalarized objective.
The constraint~\eqref{Eq:35_3} specifies the initial and final location based on $t$.

The maximum path loss of each user, \( PL_{d_{i  \leftarrow j}}^{\rm{max}} \), is a fixed value in the system. Using this maximum path loss, we calculate the optimal angle \( \theta _{j}^{\rm opt} \) by solving the non-linear partial differential equation \( \frac{\partial r_{i  \leftarrow j}}{\partial \theta _{j}^{\rm opt}} = 0 \) from the equation~\eqref{Eq:09_3}, as expressed in \cite{al2014optimal}.
    \begin{align} 
	   {\frac{{\rm \pi }\tan \theta_{j} }{9\ln \left(10\right)} +\frac{abA\exp \left(- b\left[\frac{180}{{\rm \pi }} \theta_{j} -a\right]\right)}{\left(a\exp 
          \left(-b\left[\frac{180}{{\rm \pi }} \theta_{j} -a\right]\right)+1\right)^{2} } =0}.
	   \label{Eq:17_3} 
    \end{align} 
   With the optimal angle obtained \( \theta _{j}^{\rm opt} \), if the altitude of the \( \mathbf{U}_j \)-th UAVr, denoted as \( h^{\rm Ur}_j \), is given, the corresponding coverage radius \( R_j \) can be derived using the relationship.
    \begin{align} 
	   \theta_{j}=\tan^{-1}\left(h^{\rm Ur}_j/R_j\right),
	   \label{Eq:18_3} 
    \end{align} 
    and \( \theta_{j} \) is set to the optimal angle \( \theta _{j}^{\rm opt} \). Since the maximum altitude \( h_{\max} \) is a predefined constraint, we can use the equation~\eqref{Eq:18_3} to determine the maximum coverage radius \( R_{\max} \) provided by a UAVr. Additionally, the maximum allowable LoS distance (Euclidean distance) for the \( \mathbf{U}_j \)-th UAVr will be
    \begin{align} 
	   d_{\max} = R_{\max}\sec\theta _{j}^{{\rm opt}}.
	   \label{Eq:19_3} 
    \end{align} 
   %
   \subsection{The NP-Hardness}
   \begin{equation}
        \max_{p_{i  \leftarrow j}(t)} E^{\text{Tot}} \;
        \label{Eq:obj_P2_3}
    \end{equation} 
    \text{s.t.}~\eqref{eq: SNR_max_1} --~\eqref{Eq:14_3}, ~\eqref{Eq:30_3},~\eqref{Eq:31_3},~\eqref{Eq:32_3},~\eqref{Eq:33_3}, ~\eqref{Eq:34_3}~\eqref{Eq:35_3}. \\
    %
    Note: The constraints discussed are~\eqref{Eq:obj_P1_3}. Using the above statement, we deduce the following theorem. 
    \begin{thm}\label{thm:1}
    	The original problem~\eqref{Eq:obj_P1_3} is NP-hard.
    \end{thm}
    \begin{proof}
    By relaxing the constraints of user association and assuming that all multiple UAVr serve a fixed set of users, problem~\eqref{Eq:obj_P2_3} has been proved NP-hard, as shown in \cite{7070670}. Thus, the NP-hardness of the problem~\eqref{Eq:obj_P2_3} implies that the original problem in equation~\eqref{Eq:obj_P1_3} is also NP-hard.
    The NP-hardness of the original problem in equation~(33) follows by reduction to a known NP-hard problem~\cite{7070670}, with the relaxed form in equation~(38) preserving the core combinatorial difficulty and guiding the heuristic design in Section~V.
    \end{proof}

However, optimizing transmission power with a fixed user association becomes significantly more manageable \cite{7070670}. Therefore, in the following sections, we introduce an algorithmic approach based on AMUD to optimize the transmission power of UAVr \( p_{i  \leftarrow j} \) by fixing the association variable \( \delta_{i \leftarrow j} \).
\begin{figure}[!ht]
    \vspace{-10 pt}
	\centering
	\includegraphics[width=\columnwidth]{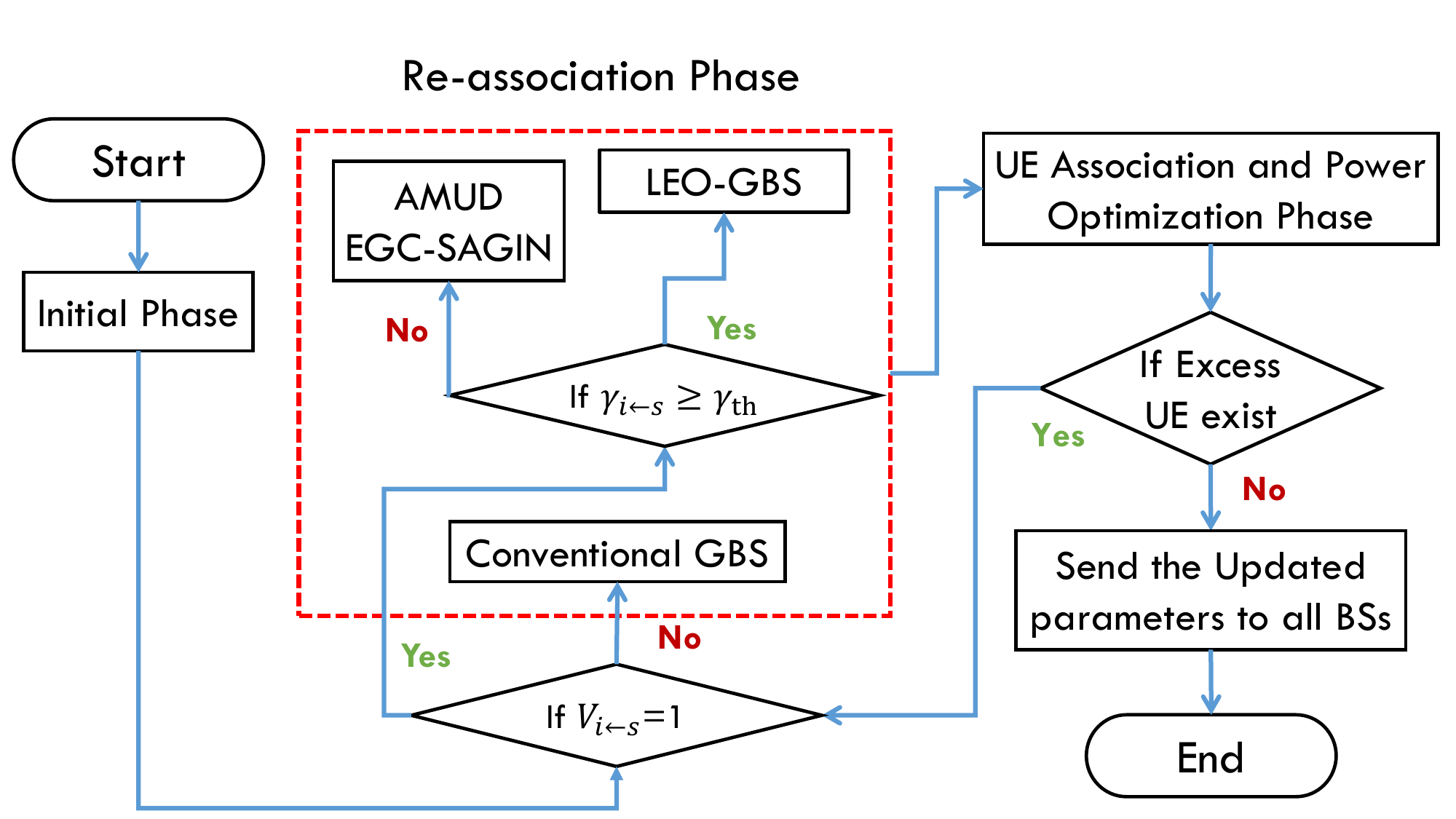}
	\caption{ \small The flowchart of the AMUD framework.}
	\label{The flowchart of the AMUD}
\end{figure}
%
%

\begin{figure*}[!ht]
	%
	\centering	
    \vspace{-8 pt}
	\hspace{-10 pt}
	\subfigure[]{
		\label{fig:fig2:a}               
		\includegraphics[width=.25\textwidth]{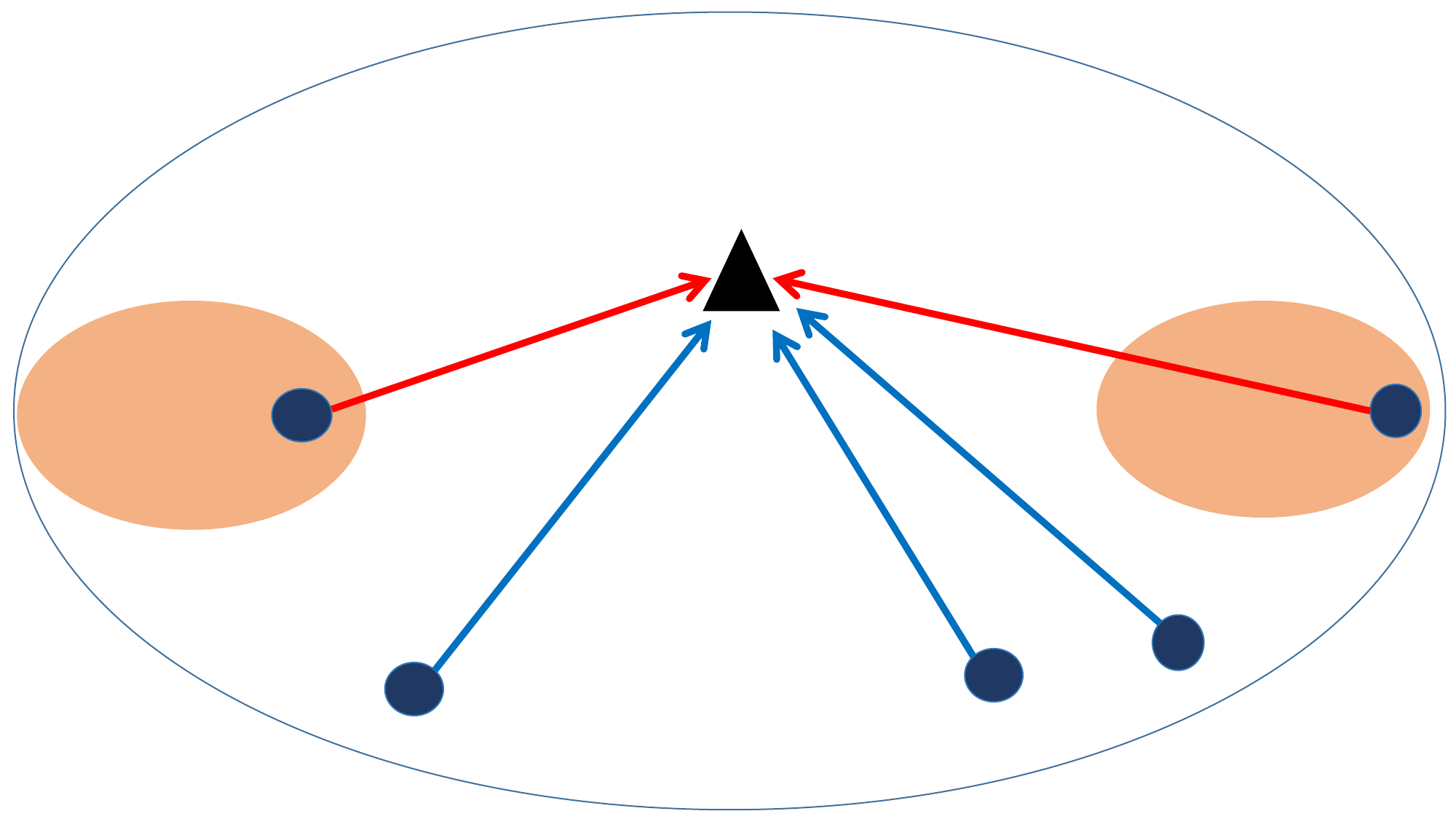}}\hspace{1.1em}%
	\subfigure[]{
		\label{fig:fig2:b}               
		\includegraphics[width=.25\textwidth]{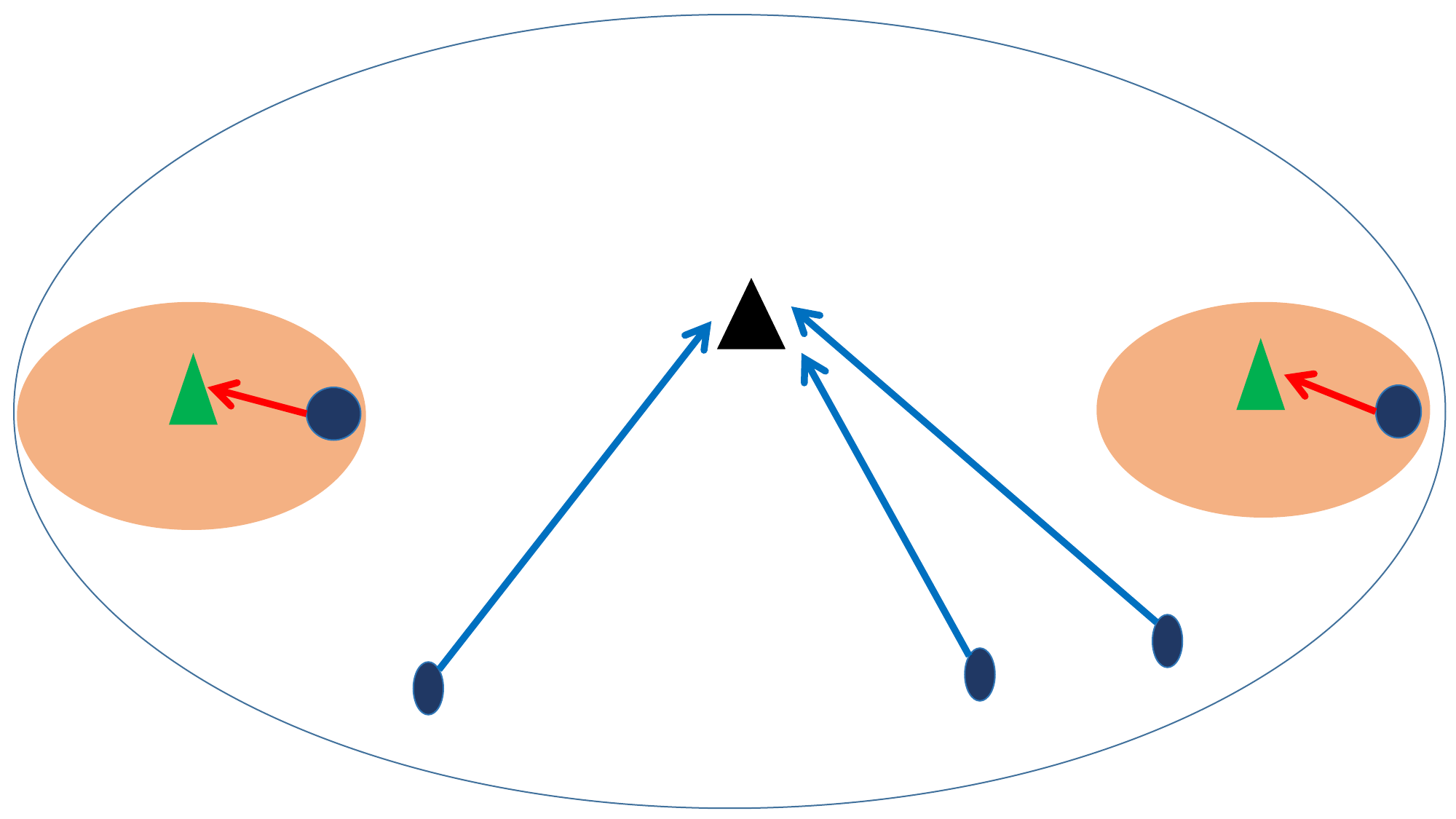}}\hspace{1em}%
	\subfigure[]{
		\label{fig:fig2:c}               
		\includegraphics[width=.35\textwidth]{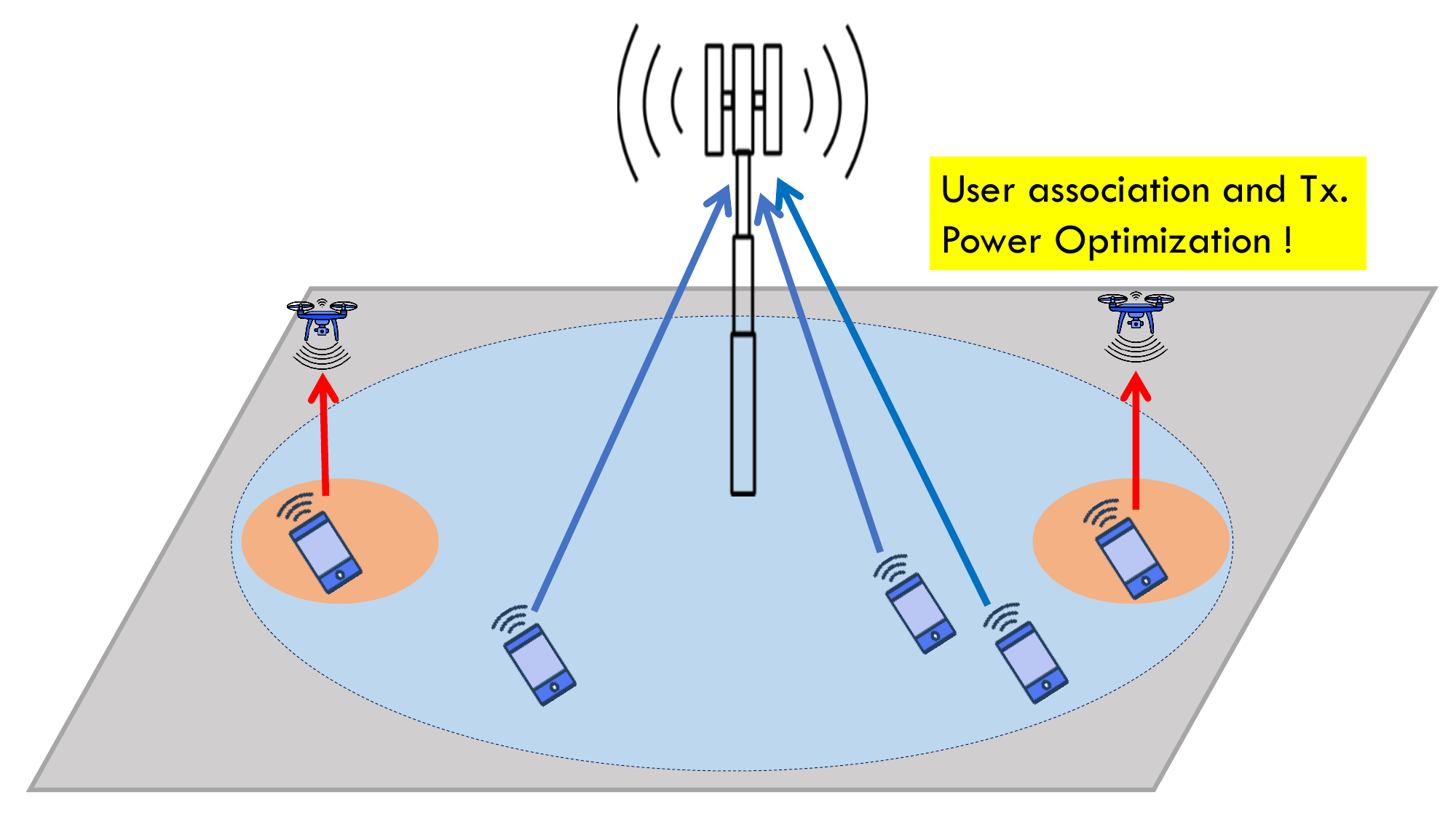}}%
	\caption{ \small The AMUD framework is visualized in $2$D and $3$D, showcasing phases: (a) Initial phase, (b) Re-association phase, and (c) User association and power optimization phase. Users are depicted as blue dots, GBS centers as black triangles, UAVr centers as green triangles, hotspots as orange-shaded areas, GBS coverage as blue circles, and UAVr coverage as orange-shaded circles. Arrows indicate user associations. For small-scale instances where exhaustive search is computationally feasible, the proposed heuristic achieves results within 5\% of the optimal solution in tested scenarios, confirming its near-optimal performance in representative conditions.}
	\label{fig:main_idea_3}                     
\end{figure*}

%
%
\section{ The Proposed AMUD Framework} 
\label{The Proposed AMUD Framework}
This section presents the main concept and an overview of the proposed AMUD method. Following this, we will describe the AMUD procedure in detail. Finally, we will discuss the key aspects of our proposed design.
%
%
\subsection{Main Idea}

The proposed AMUD framework maximizes energy efficiency~\eqref{Eq:obj_P1_3} while ensuring user fairness~\eqref{Eq:28_3} and satisfying minimum data rate constraints~\eqref{Eq:14_3} for all users. Enhance network capacity and energy efficiency through the cooperative deployment of multiple UAVr with LEO satellite support.
The design adopts a sequential and opportunistic trigger strategy, prioritizing critical issues such as NLoS mitigation and uneven user distribution before power optimization, thereby avoiding redundant actions and improving resource efficiency. The decision criteria are derived from~\eqref{Eq:21_3}, \eqref{Eq:23_3}, and \eqref{Eq:28_3}–\eqref{Eq:obj_P1_3}, ensuring a first-principles optimization framework. Extensive simulations under varying traffic, congestion, and satellite visibility conditions confirm the robustness and stability of the approach.
The AMUD framework operates in three sequential phases, as illustrated in Fig.~\ref{The flowchart of the AMUD}:

\begin{enumerate}
    \item \textbf{Initial Phase:} The CCS initializes system parameters, including user locations and GBS association states, and evaluates network load by identifying surplus GBS users \((\Omega^{\mathrm{Ex}}_{G})\) and UAVr sector densities \((\omega^{\mathrm{max}}_{j})\). This step is performed once before iteration (Fig.~\ref{fig:fig2:a}).
    
    \item \textbf{Re-association Phase:} The CCS dynamically reassigns users when \(\omega^{\mathrm{max}}_{j} > 0\) or \(\Omega^{\mathrm{Ex}}_{G} > 0\), subject to satellite visibility:
    \begin{itemize}
        \item \textbf{AMUD-SAGIN:} Activated when \(V_{i \leftarrow s} = 1\) and \(\gamma_{i \leftarrow s} < \gamma_{\text{th}}\).
        \item \textbf{EGC-SAGIN:} Equal-gain combining-based SAGIN (EGC–SAGIN) applied under the same condition to improve reliability via diversity.
        \item \textbf{LEO-GBS only:} LEO satellites and GBS (LEO-GBS),  used when \(V_{i \leftarrow s} = 1\) and \(\gamma_{i\leftarrow s} \geq \gamma_{\text{th}}\).
        \item \textbf{GBS-only:} Selected when \(V_{i \leftarrow s} \neq 1\).
    \end{itemize}
    This strategy ensures reliable connectivity and load balancing based on link conditions (Fig.~\ref{fig:fig2:b}), while leveraging \textit{a priori} decision logic to avoid unnecessary UAVr deployment.
    
    \item \textbf{User Association and Power Optimization Phase:} The CCS jointly optimizes user association and UAVr transmit power in an iterative manner until all surplus users are reassigned and optimal power allocation is achieved. The final configuration is then deployed to the UAVr nodes (Fig.~\ref{fig:fig2:c}).
\end{enumerate}
%
%
%
\begin{algorithm}[!ht]
\caption{AMUD Framework Workflow } 
\label{alg: AMUD}
\small
\textbf{Input:} \\
$\omega^{\mathrm{\rm max}}_{j}$, $\omega^{\mathrm{\rm max}}_{G}$: max user associations of UAVr and GBS, respectively; \\
$h^{\rm Ur}_j$: UAVr altitude~\eqref{Eq:02_3}; \\
${\bf E}=\{u_i\}_{i=1}^{N}$: set of horizontal user locations; where $1\leq \forall {i} \leq {N}$; \\
	 ${\rm{\bf U}} = \left\{ {{U_1},\dots,{\mathbf{U}_j},\dots ,{U_K}} \right\}$: the set of UAVr' horizontal locations,; \\	
	 ${\rm\bf\Omega}=\{{\rm\bf\Omega}_1,\dots,{\rm\bf\Omega}_K\}$: UAVr' association mapping;\\
	 ${\rm\bf\Omega_{j}}$: the set of users associated with the $j$-th UAVr;\\
	 $|{\rm\bf\Omega_{j}}|$: the number of         users associated with the $j$-th UAVr;\\
     $\bf{\Omega}^{\mathrm{\rm \bf{Ex}}}_{j}$ = $(|\Omega_{j}|-\omega^{\mathrm{\rm max}}_{j})$: the excess user associated with UAVr;  \\		
     ${\rm\bf\Omega_{G}}$: the set of users associated with the GBS;\\
     $\bf\Omega^{\mathrm{\rm Ex}}_{G}$=$(\Omega_{G}-\omega^{\mathrm{\rm max}}_{G})$: the excess user associated with GBS;  \\
     $\bf{\textit{D}^{\mathrm{\rm th}}_{j}} =\frac{{\rm\Omega_{j}}} {\pi*R^{{\rm 2}}_{j}}$: the threshold user density of UAVr;  \\
     ${\rm\bf R_{j}}$: the coverage radius of UAVr~\eqref{Eq:18_3};\\
	\textbf{Pseudo-code:}
	\begin{algorithmic}[1]
    \STATE Let $\text{P}^{\mathrm{\rm Hov}}_{j}$~\eqref{Eq:01_3} be the required hovering power of each UAVr;     \label{alg: AMUD:01}
    \STATE Let the GBS coverage area divide into multiple sub-sectors having radius $R_j$~\eqref{Eq:18_3};                \label{alg: AMUD:02}
     \STATE {Let ${\rm\bf P}=\{p_{1,1},\dots,p_{i,j},\dots,p_{N,K}\}$ be the set of transmission power from the $j$-th UAVr to the $i$-th user; 	\label{alg:AMUD_power} }
    \STATE We consider two types of hotspots, \\          \label{alg: AMUD:03}
    a) density-based,~{$\omega^{\mathrm{\rm max}}_{j} = D^{\mathrm{\rm th}}_{j}*\pi*R^{\mathrm{\rm 2}}_{j}>0$},\\
    b) number of user-based,~${\Omega}^{\mathrm{\rm {Ex}}}_{G}$=$(\Omega_{G}-\omega^{\mathrm{\rm max}}_{G})>0$;  
    \FOR{$j=1$ to $K$}                    \label{alg: AMUD:04}
	\WHILE {(${\Omega^{\mathrm{\rm Ex}}_{G}}>0 \vee \Omega^{\mathrm{\rm Ex}}_{j}>0) \wedge {V_{j  \leftarrow s}}=1 \wedge \gamma_{i \leftarrow s} < \gamma_{\rm {th}}$)}      \label{alg: AMUD:05}	
    \STATE Selects the hotspot users $i^{*}$ from the $\Omega_{G}$;\\ 
    \label{alg: AMUD:06}
    Now CCS sends UAVr $U^{\mathrm{\rm *}}_{j}$ to serve hotspot users ${i}^{*}$ \\ 
    */ Start the re-association process, and the transmission power optimization starts */  
    \FOR {$j=1$ to $K$}                     \label{alg: AMUD:07}
\FOR{$i={i}$ to $|\rm\bf\Omega_j|$}         \label{alg: AMUD:08}
\STATE Update ${\rm\bf d}_{{i}^{*} 
 \leftarrow j^{*}}= \sqrt{{\rm{{\bf r}}}^2_{{i}^{*}  \leftarrow j^{*}}+{(h^{\mathrm{\rm Ur}}_{j}})^2}$; 
    \label{alg: AMUD:09}
	\STATE Determine the average path loss between the selected user and UAVr, $PL_{h^{\rm Ur}_j,r_{{i}^{*} \leftarrow j^{*}}}^{{\rm Avg}}$, by~\eqref{Eq:09_3} with ${\rm\bf d}_{{i}^{*} \leftarrow j^{*}}$; \label{alg: AMUD:10}
\STATE Determine the minimum required transmit power $p^{\min}_{i^{*}\leftarrow j^{*}}$ for guaranteeing the SINR value~\eqref{Eq:10_3} of the $i$-th user by
{\(p^{\min}_{i^\ast\leftarrow j^\ast}
=\gamma_{\text{th}}\,(B_{i^\ast\leftarrow j^\ast}\sigma^{2}+I_u)/\|\mathbf{h}_{i^\ast\leftarrow j^\ast \leftarrow s}(t)\|^{2}\).}
\label{alg:AMUD:11}
	\STATE Determine the optimal transmit power $p_{i^{*} \leftarrow j^{*}}$ to maximize the { \(E^{\mathrm{Tot}}\) from equation~(\ref{eq: energy_efficiency})} ;     \label{alg: AMUD:12}\\
	/* Check the transmit power constraint to make sure $p_{i^{*} \leftarrow j^{*}}$ is reasonable, and update $p_{i^{*} \leftarrow j^{*}}$. */
	\IF{$p_{i^{*} \leftarrow j^{*}}^{\min} \leq {p_{\max}}$}      \label{alg: AMUD:13}
	\IF{$p_{i^{*} \leftarrow j^{*}} \leq p_{i^{*},j^{*}}^{\min}$} \label{alg: AMUD:14}
	\STATE $p_{i^{*} \leftarrow j^{*}} = p_{i^{*},j^{*}}^{\min}$; \label{alg: AMUD:15}
	\ELSIF{$p_{i^{*} \leftarrow j^{*}} \geq {p_{\max}}$}          \label{alg: AMUD:16}
	\STATE $p_{i^{*} \leftarrow j^{*}} = {p_{\max}}$;             \label{alg: AMUD:17}
	\ENDIF                                             \label{alg: AMUD:18}
	\ELSE                                              \label{alg: AMUD:19}
	\STATE $p_{i^{*} \leftarrow j^{*}} = {p_{\max}}$;             \label{alg: AMUD:20}
	\ENDIF  \label{alg: AMUD:21}			
	\STATE ${\rm{{\bf P}}}_{i^{*} \leftarrow j^{*}} = {p_{i^{*} \leftarrow j^{*}}}$; \label{alg: AMUD:22} \hfill // Commit the updated ${p_{i^{*} \leftarrow j^{*}}}$
\ENDFOR                                                \label{alg: AMUD:23}
\ENDFOR                                                \label{alg: AMUD:24}
\ENDWHILE                                              \label{alg: AMUD:25}
\ENDFOR                                                \label{alg: AMUD:26}\\
    %
	  \STATE Check if the hotspot is moving, then deploy UAVr adaptively and maintain the collision distance between the UAVr;                 \label{alg: AMUD:27}
	\STATE send ${\rm{\bf\Omega}}$, and ${\rm \bf P}$ to all UAVr;
	\label{alg: AMUD:28}
	\end{algorithmic}
\end{algorithm}
%
%
\subsection{Proposed Algorithm}
\label{Proposed Algorithm_3}
Algorithm~\ref{alg: AMUD} outlines the pseudo-code for the main procedure of the proposed AMUD approach, executed by the CCS. A detailed explanation of each step within the AMUD procedure is provided below:

\begin{itemize}
    \item Lines~\ref{alg: AMUD:01} to~\ref{alg: AMUD:03}: These lines correspond to the initialization phase, where the CCS prepares temporary matrices to store data required for subsequent re-association and optimization decisions.
    
    \item Line~\ref{alg: AMUD:04}: The CCS evaluates whether it is overloaded by iterating through a for loop.
    
    \item Line~\ref{alg: AMUD:05}: The CCS employs a while-loop system to monitor excess users. If \(\Omega_G^{\rm Ex} > 0\) or \(\omega_j^{\rm max} > 0\), and if \(V_{j  \leftarrow s} = 1\) while \(\gamma_{i \leftarrow s} < \gamma_{\rm th}\), the CCS re-associates excess users from the overloaded GBS to the LEO satellite-assisted UAV. This process continues until conditions \(\Omega_G^{\rm Ex} = 0\) and \(\omega_j^{\rm Max} = 0\) are satisfied.
    
    \item Line~\ref{alg: AMUD:06}: The CCS selects \(i^{*}\) from \(\Omega_{G}\), identifying the sector with user-generated hotspots where \(\omega_j^{\rm Max} > 0\) and \(\Omega_G^{\rm Ex} > 0\). A corresponding UAV is then activated and deployed to serve the selected users \(i^{*}\) with CCS support.
    
    \item Lines~\ref{alg: AMUD:07} to~\ref{alg: AMUD:26}: After re-associating users, a for-loop optimizes the transmit power of each UAVr with the assistance of the CCS.
    
    \item Line~\ref{alg: AMUD:09}: The Euclidean distance is calculated between the \(j\)-th UAVr and the \(i\)-th user.
    
    \item Line~\ref{alg: AMUD:10}: The path loss between the \(j\)-th UAVr and the \(i\)-th user is updated based on the changes in distance, \(d_{i  \leftarrow j}\), using equation~\eqref{Eq:09_3}.
    
    \item Line~\ref{alg:AMUD:11}: The minimum required transmit power \(p_{i  \leftarrow j}^{\min}\) is determined to ensure the SINR value defined in equation~\eqref{Eq:10_3} for the \(i\)-th user.
    
    \item Line~\ref{alg: AMUD:12}: The optimal transmission power \(p_{i  \leftarrow j}\) is selected to maximize the {total energy efficiency \(E^{\mathrm{Tot}}\) from equation~(\ref{eq: energy_efficiency})}, ensuring that it adheres to the transmission power constraints.
    
    \item Lines~\ref{alg: AMUD:13} to~\ref{alg: AMUD:21}: These lines verify that the transmission power constraints are met, followed by updating the transmission power values \(p_{i  \leftarrow j}\) at line~\ref{alg: AMUD:22}.
    
    \item Line~\ref{alg: AMUD:27}: The CCS checks for hotspot mobility and adaptively deploys multiple UAVr, ensuring a safe collision-free distance between them.
    
    \item Final Step: The CCS transmits updated parameter sets, \(\Omega\) and \(P\), to all UAVr for deployment adjustments.
\end{itemize}

In the upcoming section, we will discuss the simulated results, along with inferences and observations.
%
%
%
\section{Results and Discussion}
\label{sec: Simulation Results and Performance Analysis_3}

This section presents the simulation results for the proposed AMUD framework, obtained by solving~\eqref{Eq:obj_P1_3}, and evaluates them in terms of total network capacity, energy efficiency, and capacity fairness. Simulations are conducted in \textsc{MATLAB} R$2020$b for an urban environment with randomly distributed users and varying GBS overload levels (excess users).
The proposed AMUD--SAGIN framework is compared with three benchmark schemes: EGC--SAGIN, only LEO--GBS, and a conventional GBS-only network. Detailed descriptions of these baselines are provided in Subsection~\ref{subsec:BenchmarkSchemes}.
To comply with UAV regulations in Taiwan, the UAVr altitude is fixed at $100$~m~\cite{laws_taiwan}. The air-to-ground channel parameters $(a, b, \eta_{\rm LoS}, \eta_{\rm NLoS}) = (9.61, 0.16, 1, 20)$ are adopted from~\cite{al2014optimal}. The system bandwidth is set to $20$~MHz, and the noise power density is $-174$~dBm / Hz~\cite{9946428}. Other parameters, including carrier frequency \(f_c\), UAVr transmit power \(p_{\max}\), and user association threshold \(\omega_{\rm max}\), follow standard SAGIN configurations and are summarized in Table~\ref{Simulation_Parameter_3}.

%
\begin{table}[!ht]
\centering
\caption{Simulation Parameters}
\label{Simulation_Parameter_3}
\begin{tabular}{|p{3cm}|c|c|}
\hline
\footnotesize
\textbf{Parameter} & \textbf{Symbol} & \textbf{Value} \\ \hline
    Maximum path loss & $PL_{d_{i \leftarrow j}}^{\mathrm{max}}$ & $119$ dB \\ \hline
    Carrier frequency & $f_c$ & 2.4 GHz \\ \hline
    LEO Satellite Altitude & $h_{\rm s}$ & $1000$ km \\ \hline
    UAVr Service Altitude & {$h^{\rm Ur}_j$} & $100$ m \\ \hline
    GBS Coverage Area & $F_{\rm hoz} \times F_{\rm ver}$ & $1 \times 1$ km \\ \hline
    Allocated Bandwidth & $B_{i \leftarrow j}, B_{i \leftarrow s}, B_{i \leftarrow G}$ & $20$ MHz \\ \hline
    Number of UAVr & $K$ & 2 \\ \hline
    Path loss exponent & $\alpha_{\rm exp}$ & 2 \\ \hline
    Total Number of Users & $N$ & $400$ \\ \hline
    Max Associations at UAVr & $\omega^{\mathrm{max}}_j$ & $200$ \\ \hline
    Max Associations at GBS & $\omega^{\mathrm{max}}_G$ & $100$ \\ \hline
    SINR Threshold & $\gamma_{\rm th}$ & 3 dB \\ \hline
    Hovering Power & $p_j^{\mathrm{Hov}}$ & 58 W \\ \hline
    UAVr Transmission Power & $p_{\rm max}$ & 20 dBm \\ \hline
    LEO Transmission Power & $P_s^{\rm tx}$ & 50 dBm \\ \hline
    GBS Transmission Power & $P_{\!G}^{\mathrm tx}$ & 40 dBm \\ \hline
\end{tabular}
\end{table}

\subsection{Benchmark Schemes for Comparison}
\label{subsec:BenchmarkSchemes}
To evaluate the performance of the proposed AMUD framework, we compare it with the following three benchmark schemes under identical simulation settings:
\begin{itemize}
    \item \textbf{AMUD--SAGIN:} This scheme combines the SINR of the satellite–user and UAVr–user links using a weighted approach that prioritizes stronger links, as described in Equation~\eqref{eq: SNR_max_1}. Multiple UAVr are dynamically positioned based on user congestion. 
    \item \textbf{EGC--SAGIN:} This scheme uses an equal-gain combination, where the SINR of the satellite and UAVr links are averaged equally. Multiple UAVr are deployed at the same locations as in AMUD. The only difference from AMUD lies in the SINR combining method.
    \item \textbf{LEO--GBS:} This setup involves only the LEO satellite and GBS. The satellite helps users who cannot be served by GBS when it is visible (Fig.~\ref{The flowchart of the AMUD}), and the SINR condition in Equation~\eqref{eq: Satellite_User} is satisfied, and no UAVr are used.
    \item \textbf{GBS--Only:}  
    A purely terrestrial network where users are served only by GBS. Users exceeding the GBS capacity are dropped. Neither satellite assistance nor UAVr is used.
\end{itemize}

%
\begin{figure}[!ht]
	\centering
	\includegraphics[width=0.48\textwidth, height=6.2cm]{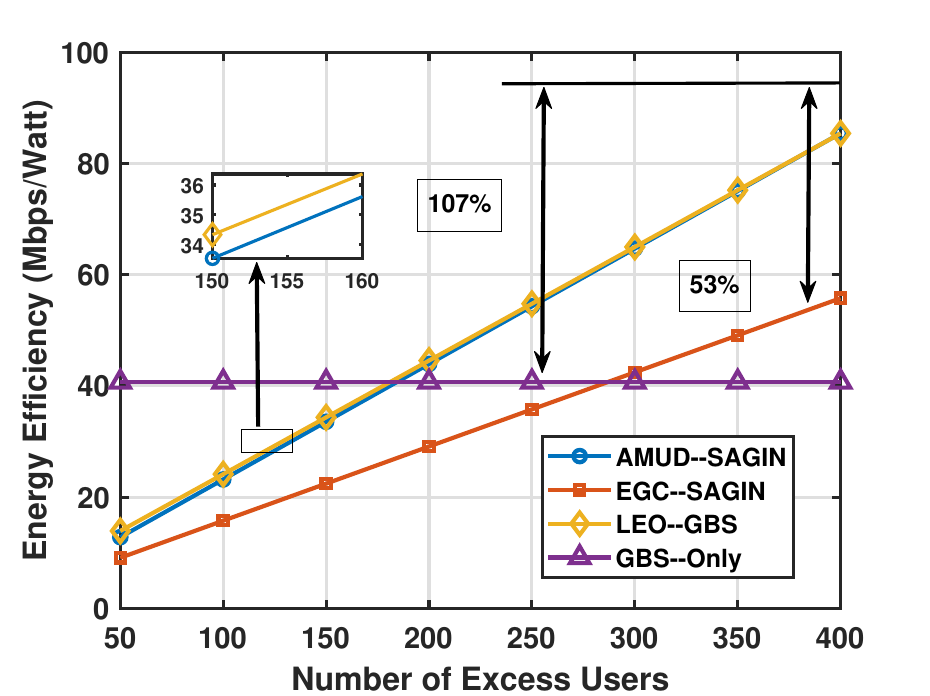}
	\caption{ {\small Performance comparison of the proposed AMUD--SAGIN with other approaches, including EGC--SAGIN, LEO--GBS, and GBS--only, in terms of energy efficiency with varying numbers of excess users.}}
	\label{EE_uav_Tx_and_UE_vary}
\end{figure}

\subsection{Energy Efficiency Analysis}
Fig.~\ref{EE_uav_Tx_and_UE_vary} evaluates the energy efficiency of the proposed AMUD framework under varying user loads ($50{:}400$). At low user loads, AMUD--SAGIN exhibits slightly lower efficiency than the GBS-only system due to the initial overhead of UAVr deployment and coordination, although this overhead contributes less than $2\%$ additional signaling load.
As the user load increases ($200$--$400$), AMUD--SAGIN significantly outperforms GBS-only, achieving up to $107\%$ higher energy efficiency in $400$ users. This improvement is driven by dynamic resource allocation and cooperative utilization of UAVr and LEO satellites to mitigate congestion. Moreover, AMUD--SAGIN surpasses EGC--SAGIN by $53\%$, demonstrating superior adaptability under high-demand conditions.
Although the energy efficiency trends of AMUD--SAGIN and LEO--GBS appear similar, AMUD--SAGIN consistently outperforms LEO--GBS due to its adaptive link-weighting mechanism, which prioritizes stronger links and enhances overall capacity.

\begin{figure}[!ht]
	\centering
	\includegraphics[width=0.48\textwidth, height=6.2cm]{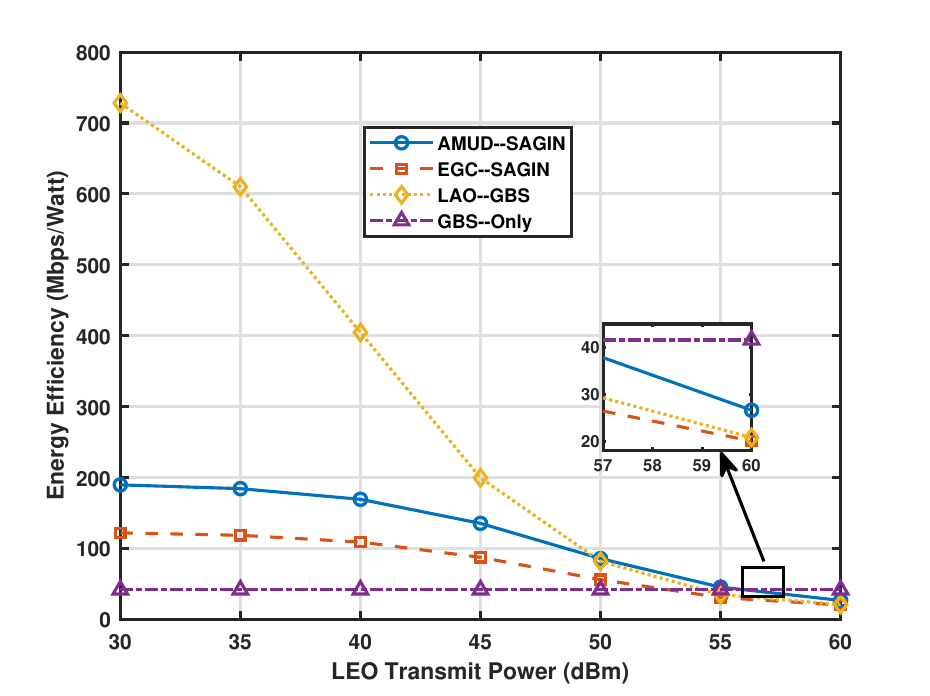}
	\caption{\small Performance comparison of the proposed AMUD--SAGIN with other approaches, including EGC--SAGIN, LEO--GBS, and GBS--only, in terms of energy efficiency with varying LEO satellite transmission power levels.}
	\label{Final_EE_leoTx_vary}
\end{figure}
The simulation (Fig.~\ref{Final_EE_leoTx_vary}) evaluates the energy efficiency of the proposed AMUD--SAGIN framework under varying LEO satellite transmission powers ($30$--$60$~dBm) in a SAGIN. When the transmit power of the LEO satellite is low ($30$~ dBm), the energy efficiency of AMUD--SAGIN is slightly lower than that of the LEO--GBS system due to the satellite's limited transmission power. However, it still outperforms EGC--SAGIN by $56\%$ and GBS only by $360\%$, showcasing its ability to leverage UAVr-LEO collaboration even at lower power levels. As the transmit power increases to $60$~dBm, AMUD--SAGIN significantly improves its performance, achieving $30\%$ higher energy efficiency than SAGIN and $27\%$ higher than EGC--SAGIN. 

\begin{figure}[!ht]
	\centering
	\includegraphics[width=0.48\textwidth, height=6.2cm]{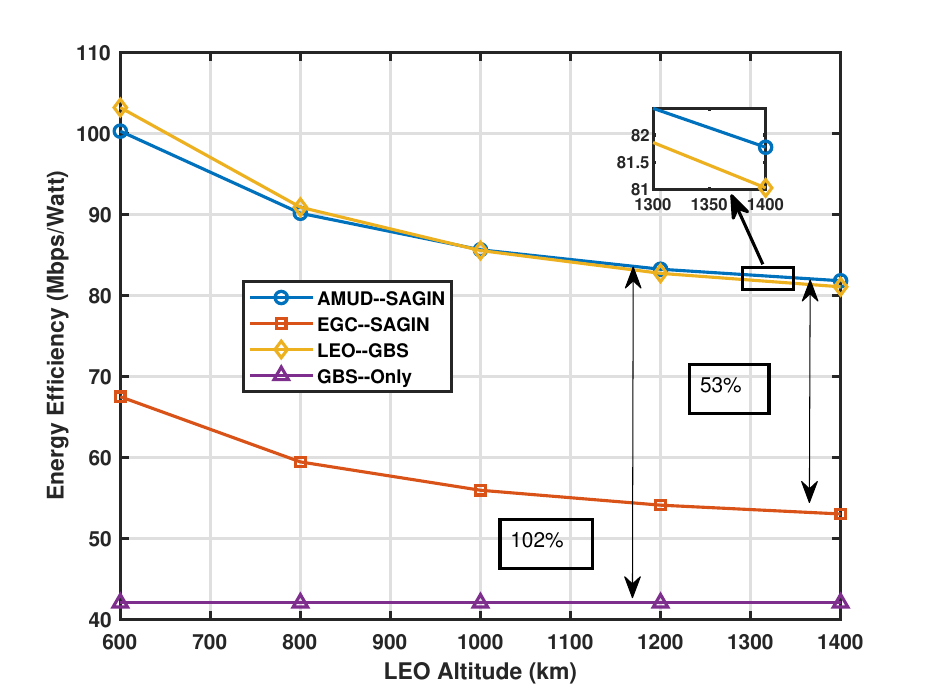}
    \caption{\small Performance comparison of the proposed AMUD--SAGIN with other approaches, including EGC--SAGIN, LEO--GBS, and GBS--only, in terms of energy efficiency with varying LEO satellite altitudes.}
	\label{Final_EE_leo_Ht_vary}
\end{figure}

Fig.~\ref{Final_EE_leo_Ht_vary} illustrates the energy efficiency (Mbps/Watt) as a function of the altitude of the LEO satellite ($600$--$1400$~km) with excess users $400$. The proposed AMUD--SAGIN approach achieves the highest energy efficiency at $81.74$~Mbps/Watt, outperforming LEO--GBS by $2\%$, EGC--SAGIN by $53\%$ ($52.95$~Mbps / Watt), and the baseline system only GBS by $102\%$ ($41.44$~Mbps/Watt). As LEO altitude increases, energy efficiency generally decreases due to higher path loss and increased transmission power requirements. However, AMUD--SAGIN consistently maintains superior performance at varying altitudes, demonstrating its robustness in managing energy consumption through optimized UAVr deployment.

%
\begin{figure}[!ht]
	\centering
	\includegraphics[width=0.48\textwidth, height=6.2cm]{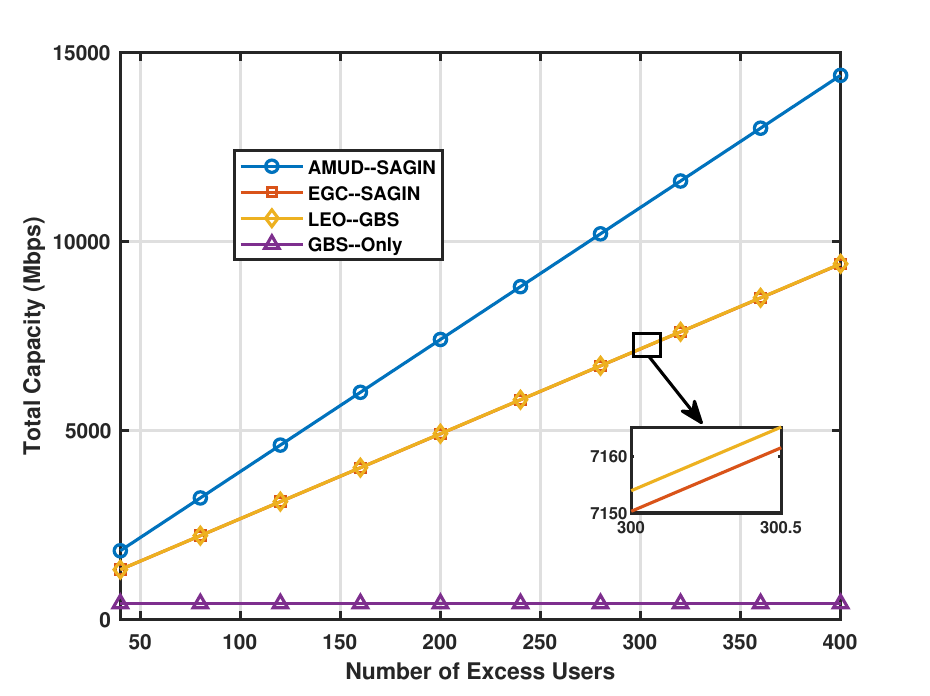}
    \caption{\small Performance comparison of the proposed AMUD--SAGIN with other approaches, including EGC--SAGIN, LEO--GBS, and GBS--only, in terms of total capacity with varying excess user counts.}
    \label{Final_Cap_user_vary}
\end{figure}

\subsection{ Network Capacity Analysis}
Fig.~\ref{Final_Cap_user_vary} shows the total capacity (Mbps) as a function of the excess users in various communication schemes~({$40{:}400$}). The proposed AMUD approach consistently outperforms all other schemes, achieving a peak capacity of $14,386$ Mbps for 400 excess users, highlighting its superior scalability and efficiency. In contrast, the LEO-GBS and EGC-SAGIN schemes demonstrate comparable performance, reaching approximately $9,401$ Mbps and $9,397$ Mbps, respectively. This indicates that while both benefit from satellite-ground integration, they lack the dynamic relay-deployment capabilities inherent in AMUD-SAGIN. The GBS-only scheme records the lowest capacity, at $419$ Mbps, due to its reliance solely on terrestrial infrastructure. These findings validate the effectiveness of AMUD-SAGIN in utilizing UAVr to significantly improve network capacity, positioning it as a robust solution for managing high user densities in SAGINs.

\begin{figure}[!ht]
	\centering
	\includegraphics[width=0.48\textwidth, height=6.2cm]{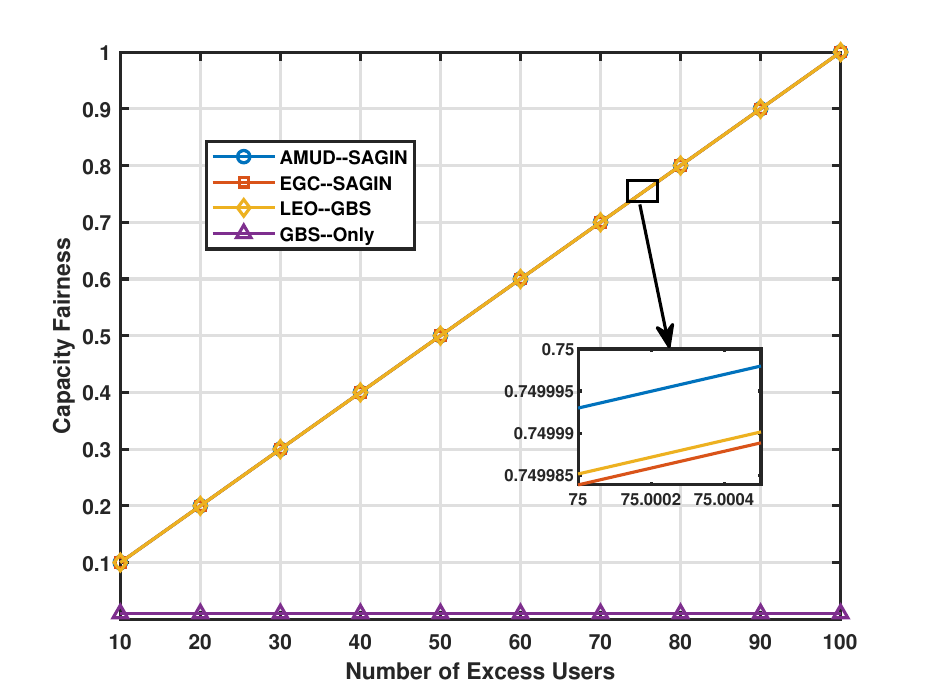}
    \caption{\small Performance comparison of the proposed AMUD--SAGIN with other approaches, including EGC--SAGIN, LEO--GBS, and GBS--only, in terms of capacity fairness with varying excess user counts.}
    \label{Final_Cap_Fairness}
\end{figure}

 \subsection{ The Fairness Analysis}
Fig.~\ref{Final_Cap_Fairness} presents the capacity fairness under different network configurations~($10{:}100$). The proposed AMUD framework, along with EGC--SAGIN and LEO--GBS, achieves near-optimal fairness with Jain’s index close to $1.0$, indicating an almost uniform distribution of capacity among users. In particular, AMUD--SAGIN, EGC--SAGIN, and LEO-GBS exhibit nearly identical fairness performance.
The inclusion of LEO--GBS already mitigates severe fairness issues by providing connectivity to users that would otherwise be dropped in a congested GBS-only network. As a result, all SAGIN-based schemes achieve near-ideal fairness.
In contrast, the GBS-only configuration shows extremely poor fairness (index $\approx 0.01$), highlighting its inability to distribute resources equitably. These results demonstrate that while AMUD primarily enhances capacity and energy efficiency, it preserves the high fairness achieved by SAGIN-based architectures, emphasizing the necessity of UAVr and satellite integration for equitable resource allocation.
%

\section{ Conclusion } 
\label{Sec: Conclusion_3}
This paper presents the AMUD strategy for integrated satellite air-ground networks. The AMUD framework effectively utilizes multiple UAVr within a LEO satellite-based amplify-and-forward system to enhance the signal-to-interference-plus-noise ratio at the user level. It addresses critical challenges that traditional GBS faces, including increased user density, non-line-of-sight conditions, signal blockages, and difficulty in managing random access users.
The AMUD framework emphasizes user fairness in resource allocation while avoiding collisions in multi-UAVr scenarios. In response to dynamic user traffic variations, it adaptively deploys multiple UAVr to maximize user SINR through cooperative diversity-based signal combinations. This is achieved through intelligent communication between LEO satellites, multiple UAVr, and users. 
The simulation results indicate that the AMUD framework significantly enhances network performance metrics, including total network capacity, energy efficiency, and capacity fairness, particularly under conditions of high user density. These findings underscore the potential of integrating multiple UAVr with LEO satellite technologies to meet the demands of future urban communication networks, especially in scenarios characterized by fluctuating user demands and challenging environmental conditions.
%

%
%
\appendices
\section{Proof of Signal Combiner}
\label{appendix-A_3}
\begin{proof}
The weighted combination of the output ~\cite{jagannatham2015principles} at the receiver, derived from~\eqref{Eq: receive signal}, is expressed as,
\begin{align}
    \gamma_{\text{AF}}(\mathbf{w}) = {\text{Signal Power}}/{\text{Noise Power}}.
    \label{appendix-A1_3}
\end{align}
\noindent{The signal power is calculated as,}
\begin{align}
\mathbb{E}\!\left\{\left(\mathbf{w}^{\dagger}\mathbf{h}\right)\mathbf{x}_{\text{sym}}\mathbf{x}_{\text{sym}}^{\dagger}\left(\mathbf{w}\mathbf{h}^{\dagger}\right)\right\}
&=  P_s^{\text{tx}}\,\left(\mathbf{w}^{\dagger}\mathbf{h}\,\mathbf{w}\,\mathbf{h}^{\dagger}\right)\label{appendix-A2_3}
\end{align}
where ${P_{s}^{\text{tx}}}=\mathbb{E} \left\{ \mathbf{x}_{\text{sym}} \mathbf{x}_{\text{sym}}^{\dagger} \right\}$. The noise power is calculated as,
%
\begin{align}
    \mathbb{E}\left\{ \left( \mathbf{w}^{\dagger} {{\textbf{n}}} \right) \left( \mathbf{w}^{\dagger} {{\textbf{n}}} \right)^{\dagger} \right\} 
    = \mathbb{E}\left\{ \mathbf{w}^{\dagger} {{\textbf{n}}} {{\textbf{n}}}^{\dagger} \mathbf{w} \right\} = \mathbf{w}^{\dagger} \mathbb{E} \left\{ {{\textbf{n}}} {{\textbf{n}}}^{\dagger} \right\} \mathbf{w}.
    \label{appendix-A3_3}
\end{align}

Substituting \eqref{appendix-A2_3} and \eqref{appendix-A3_3} into \eqref{appendix-A1_3}, we obtain:
\begin{align}
    \gamma_{\text{AF}}(\mathbf{w}) 
    &= \frac{{P_s^{\text{tx}}} \mathbf{w}^{\dagger} {\textbf{h}} {\textbf{h}}^{\dagger} \mathbf{w}}{\mathbf{w}^{\dagger} \mathbb{E} \left\{ {{\textbf{n}}} {{\textbf{n}}}^{\dagger} \right\} \mathbf{w}} 
    = {P_{s}^{\text{tx}}} \frac{\mathbf{w}^{\dagger} R_{\rm k} \mathbf{w}}{\mathbf{w}^{\dagger} \mathscr{R}_{\rm n} \mathbf{w}},
    \label{signal_combiner_final}
\end{align}
where $R_{\rm k} = {\textbf{h}}{\textbf{h}}^{\dagger}, \text{and}~ \mathscr{R}_{\rm n}= \mathbb{E}\{ {{\textbf{n}}} {{\textbf{n}}}^{\dagger} \} \\
    =  \mathbb{E} \left\{ \begin{bmatrix} \mathbf{n}_{ i\leftarrow s} \\ \mathbf{\mathbf{h}_{ i \leftarrow j \leftarrow s}} \mathcal{G} {n}_{ j \leftarrow s} + \mathbf{n}_{i \leftarrow j \leftarrow s} \end{bmatrix} \begin{bmatrix} \mathbf{n}_{ i\leftarrow s} \\ 
    \mathbf{\mathbf{h}_{ i \leftarrow j \leftarrow s}} \mathcal{G} {n}_{ j \leftarrow s} + \mathbf{n}_{i \leftarrow j \leftarrow s} \end{bmatrix}^{\dagger} \right\} \nonumber\\
    ={{\sigma }^2}\begin{bmatrix}
   {1} & 0  \\
   0 & {1}+{{\mathbf{h}_{ i \leftarrow j \leftarrow s}}}\mathbf{h}_{ i \leftarrow j \leftarrow s}^{\dagger }{{\mathcal{G}}^{2}} \nonumber \\
    \end{bmatrix} ={{\sigma }^2}\left( {\mathfrak{I}_L}\oplus {{R}_{\text{M}}} \right) .\nonumber\\ 
$
\end{proof}
%
%
\section{Proof of Optimal Weight Vector and SINR}
\label{appendix-B_3}
Let \( u = \mathbf{w}^{\dagger} R_{\rm k} \mathbf{w} \) and \( v = \mathbf{w}^{\dagger} \mathscr{R}_{\text{n}} \mathbf{w} \) from \eqref{signal_combiner_final}, then the function $\gamma_{\text{\rm AF}}(\mathbf{w})$ can be written as
\begin{equation}
    \gamma_{\text{\rm AF}}(\mathbf{w}) = {u}/{v}
    \label{appendix-B1_3}
\end{equation}
Using the quotient rule for derivatives:
\[
\frac{d}{d\mathbf{w}} \left( \frac{u}{v} \right) = \frac{v \frac{d}{d\mathbf{w}}(u) - u \frac{d}{d\mathbf{w}}(v)}{v^2}
\]

The derivatives of \( u \) and \( v \) are:
\[
{d}/{d\mathbf{w}} (u) = 2 R_k \mathbf{w}, \quad {d}/{d\mathbf{w}} (v) = 2 R_n \mathbf{w}
\]

Substituting these into the quotient rule:
\[
\frac{d}{d\mathbf{w}} \gamma_{\text{AF}}(\mathbf{w}) = \frac{2 (\mathbf{w}^\dagger R_n \mathbf{w}) R_k \mathbf{w} - 2 (\mathbf{w}^\dagger R_k \mathbf{w}) R_n \mathbf{w}}{(\mathbf{w}^\dagger R_n \mathbf{w})^2}
\]

Setting the derivative to zero leads to the following.
\[
(\mathbf{w}^\dagger R_n \mathbf{w}) R_k \mathbf{w} = (\mathbf{w}^\dagger R_k \mathbf{w}) R_n \mathbf{w}
\]

This simplifies to
\[
R_k \mathbf{w} = \lambda R_n \mathbf{w}, \quad \text{where } \lambda = ({\mathbf{w}^\dagger R_k \mathbf{w}})/({\mathbf{w}^\dagger R_n \mathbf{w}})
\]

Rewrite:~$R_n^{-1} R_k \mathbf{w} = \lambda \mathbf{w}$. 
The optimal weight vector \( \mathbf{w}_{\text{opt}} \) is the eigenvector of \( R_n^{-1} R_k \) corresponding to its largest eigenvalue~\cite{holter2002optimal}. Let \( {\textbf{h}}' = R_n^{-1} {\textbf{h}} \). Then:
\[
\mathbf{w}_{\text{opt}} = c_r {\textbf{h}}' = c_r R_n^{-1} {\textbf{h}}, \quad c_r \neq 0
\]

Substituting \( \mathbf{w}_{\text{opt}} \) into the eigenvalue equation gives: \[ \lambda_{\text{max}} = {\textbf{h}}^\dagger R_n^{-1} {\textbf{h}} \]

Thus, the optimal weight vector maximizes \( \gamma_{\text{AF}}(\mathbf{w}) \) and is proportional to \( R_n^{-1} {\textbf{h}} \).

\subsection*{Step 1: Substitution of ${R}_{n}^{-1}$=$\left( \mathfrak{I}_L \oplus R_{\rm M} \right)^{-1}$}
$\frac{p_s^{\text{tx}}}{{\sigma}^2} \begin{bmatrix} \mathbf{h}_{i\leftarrow s} \\ \mathbf{\mathbf{h}_{i\leftarrow j \leftarrow s}} \mathcal{G} \hbar_{j\leftarrow s} \end{bmatrix}^\dagger \left( \mathfrak{I}_L \oplus R_{\text{M}}^{-1} \right) \begin{bmatrix} \mathbf{h}_{i\leftarrow s} \\ \mathbf{\mathbf{h}_{i\leftarrow j \leftarrow s}} \mathcal{G} \hbar_{j\leftarrow s} \end{bmatrix}.
$

\subsection*{Step 2: Simplifying the matrix multiplication}
The matrix \(\left( \mathfrak{I}_L \oplus R_{\text{M}}^{-1} \right)\) acts on the vector \(\begin{bmatrix} \mathbf{h}_{i\leftarrow s} \\ \mathbf{h}_{i\leftarrow s} \mathcal{G} \hbar_{j\leftarrow s} \end{bmatrix}\) as,
\[
\left( \mathfrak{I}_L \oplus R_{\text{M}}^{-1} \right) \begin{bmatrix} \mathbf{h}_{i\leftarrow s} \\ \mathbf{h}_{i\leftarrow s} \mathcal{G} \hbar_{j\leftarrow s} \end{bmatrix} = \begin{bmatrix} \mathbf{h}_{i\leftarrow s} \\ R_{\text{M}}^{-1} \cdot \mathbf{h}_{i\leftarrow s} \mathcal{G} \hbar_{j\leftarrow s} \end{bmatrix}.
\]
This yields the following. 
\[
\frac{p_s^{\text{tx}}}{{\sigma}^2} \begin{bmatrix} \mathbf{h}_{i\leftarrow s} \\ \mathbf{\mathbf{h}_{i\leftarrow s}} \mathcal{G} \hbar_{j\leftarrow s} \end{bmatrix}^\dagger \begin{bmatrix} \mathbf{\mathbf{h}_{i\leftarrow s}} \\ R_{\text{M}}^{-1} \cdot \mathbf{\mathbf{h}_{i\leftarrow s}} \mathcal{G} \hbar_{j\leftarrow s} \end{bmatrix}.
\]
{Step 3: Expanding the inner product:} 
\[
\frac{p_s^{\text{tx}}}{{\sigma}^2} \left( \mathbf{\mathbf{h}_{i\leftarrow s}}^\dagger \mathbf{\mathbf{h}_{i\leftarrow s}} + \left( \mathbf{\mathbf{h}_{i\leftarrow s}} \mathcal{G} \hbar_{j\leftarrow s} \right)^\dagger R_{\text{M}}^{-1} \cdot \left( \mathbf{\mathbf{h}_{i\leftarrow s}} \mathcal{G} \hbar_{j\leftarrow s} \right) \right).
\]
The first term is simply \( |\mathbf{\mathbf{h}_{i\leftarrow s}}|^2 \),  
and the given: $R_{\text{M}} = 1 + \mathbf{\mathbf{h}_{i\leftarrow s}} \mathbf{\mathbf{h}_{i\leftarrow s}}^\dagger \mathcal{G}^2,$ and substituting this into the second term:


\begin{equation}
\begin{aligned}
    &= \frac{p_s^{\text{tx}}}{\sigma^2} \left( \left| \mathbf{h}_{ i \leftarrow s} \right|^2 + \frac{\left| \mathbf{h}_{ i \leftarrow j \leftarrow s} \mathcal{G} \right|^2 \left| \hbar_{j\leftarrow s} \right|^2}{1 + \left| \mathbf{h}_{ i \leftarrow j \leftarrow s} \mathcal{G} \right|^2 } \right) \\
    & = \gamma_{i \leftarrow s} + \gamma_{i \leftarrow j} ({\left\| \hbar_{j\leftarrow s} \right\|^2})/({1 + \left\| \mathbf{h}_{ i \leftarrow j \leftarrow s} \mathcal{G}\right\|^2}) \\
    & \text{Multiply } \frac{p_s^{\text{tx}}}{\sigma^2} \text{ into the expression:} \\
    & = \gamma_{i \leftarrow s} + \gamma_{i \leftarrow j} \frac{\frac{p_s^{\text{tx}}}{\sigma^2} \left\| \hbar_{j\leftarrow s} \right\|^2}{\frac{p_s^{\text{tx}}}{\sigma^2} \left( 1 + \left\| \mathbf{h}_{ i \leftarrow j \leftarrow s} \mathcal{G} \right\|^2 \right)} 
     = \gamma_{i \leftarrow s} + \frac{\gamma_{i \leftarrow j} \gamma_{j \leftarrow s}}{\gamma_{i \leftarrow j} + \varsigma}.
\end{aligned}
\end{equation}
    where $\varsigma =\frac{1}{{{\sigma }^{2}}{\mathcal{G}^{2}}}$~\cite{6503767},
    \( R_{\text{M}} = \mathfrak{I}_L + \mathcal{G}^2 \mathbf{\mathbf{h}_{ i \leftarrow j \leftarrow s}} \mathbf{\mathbf{h}_{ i \leftarrow j \leftarrow s}}^\dagger \) represents the covariance matrix of \( {{\textbf{n}}} \), $\mathfrak{I}_L$ is the identity matrix $L*L$.

\vspace{-5pt}
\bibliographystyle{IEEEtran}
\bibliography{reference}

@ARTICLE{7918510,
  author={Alzenad, Mohamed and El-Keyi, Amr and Lagum, Faraj and Yanikomeroglu, Halim},
  journal={IEEE Wireless Commun. Lett.}, 
  title={{3-D} Placement of an Unmanned Aerial Vehicle Base Station ({UAV}-{BS}) for Energy-Efficient Maximal Coverage}, 
  year={2017},
  volume={6},
  number={4},
  pages={434-437},
  doi={10.1109/LWC.2017.2700840}}

@ARTICLE{8949359,
  author={Singh, Vibhum and Upadhyay, Prabhat K. and Lin, Min},
  journal={IEEE Wireless Commun. Lett.}, 
  title={On the Performance of {NOMA}-Assisted Overlay Multiuser Cognitive Satellite-Terrestrial Networks}, 
  year={2020},
  volume={9},
  number={5},
  pages={638-642},
  doi={10.1109/LWC.2020.2963981}}

@book{jagannatham2015principles,
  title={Principles of modern wireless communication systems},
  author={Jagannatham, Aditya K},
  year={2015},
  publisher={McGraw-Hill Education}
}

@INPROCEEDINGS{10757512,
  author={Bhola and Chen, Yu-Jia and Balakrishnan, Ashutosh and De, Swades and Wang, Li-Chun},
  booktitle={IEEE 100th Vehicular Technology Conference}, 
  title={Cooperative {UAV}-{Relay} based Satellite Aerial Ground Integrated Networks}, 
  year={2024},
  volume={},
  number={},
  pages={1-5},
  doi={10.1109/VTC2024-Fall63153.2024.10757512}}

@ARTICLE{9383780,
  author={Aydin, Yucel and Kurt, Gunes Karabulut and Ozdemir, Enver and Yanikomeroglu, Halim},
  journal={IEEE IoT J.}, 
  title={Group Handover for Drone Base Stations}, 
  year={2021},
  volume={8},
  number={18},
  pages={13876-13887},
  doi={10.1109/JIOT.2021.3068297}}

@article{liu2022evolution,
  title={Evolution of {NOMA} toward next-generation multiple access ({NGMA}) for {6G}},
  author={Liu, Yuanwei and Zhang, Shuowen and Mu, Xidong and Ding, Zhiguo and Schober, Robert and Al-Dhahir, Naofal and Hossain, Ekram and Shen, Xuemin},
  journal={IEEE J. Sel. Areas Commun.},
  volume={40},
  number={4},
  pages={1037--1071},
  year={2022},
  publisher={IEEE}
}

@ARTICLE{8626457,
  author={Guidotti, Alessandro and Vanelli-Coralli, Alessandro and Conti, Matteo and Andrenacci, Stefano and Chatzinotas, Symeon and Maturo, Nicola and Evans, Barry and Awoseyila, Adegbenga and Ugolini, Alessandro and Foggi, Tommaso and Gaudio, Lorenzo and Alagha, Nader and Cioni, Stefano},
  journal={IEEE Trans. Veh. Technol.}, 
  title={Architectures and Key Technical Challenges for {5G} Systems Incorporating Satellites}, 
  year={2019},
  volume={68},
  number={3},
  pages={2624-2639},
  doi={10.1109/TVT.2019.2895263}}

@ARTICLE{9275613,
  author={Giordani, Marco and Zorzi, Michele},
  journal={IEEE Netw.}, 
  title={Non-Terrestrial Networks in the {6G} Era: Challenges and Opportunities}, 
  year={2021},
  volume={35},
  number={2},
  pages={244-251},
  doi={10.1109/MNET.011.2000493}}

@ARTICLE{9755995,
  author={Chen, Yu-Jia and Chen, Wei and Ku, Meng-Lin},
  journal={IEEE Commun. Lett.}, 
  title={Trajectory Design and Link Selection in {UAV}-Assisted Hybrid Satellite-Terrestrial Network}, 
  year={2022},
  volume={26},
  number={7},
  pages={1643-1647},
  doi={10.1109/LCOMM.2022.3166961}}

@article{10.1109/COMST.2022.3197695,
author = {Al-Hraishawi, Hayder and Chougrani, Houcine and Kisseleff, Steven and Lagunas, Eva and Chatzinotas, Symeon},
journal = {Commun. Surveys Tuts.},
title = {A Survey on Nongeostationary Satellite Systems: The Communication Perspective},
year = {2023},
volume = {25},
number = {1},
pages = {101–132},
issue_date = {Firstquarter 2023},
publisher = {IEEE Press},
doi = {10.1109/COMST.2022.3197695},

}

@manual{3GPPTS22125,
    organization  = "3GPP Standard TS 22.125 V17.2.0, October",
    title         = "{3GPP} Technical Specification Group Services and System Aspects; Unmanned aerial system (UAS) support in {3GPP}",
    year          =  2020,
    month         =  "",
    number        = "",
    note          = "" }

@ARTICLE{8917591,
  author={Sharma, Pankaj K. and Deepthi, Darapaneni and Kim, Dong In},
  journal={IEEE Commun. Lett.}, 
  title={Outage Probability of {3-D} Mobile {UAV} Relaying for Hybrid Satellite-Terrestrial Networks}, 
  year={2020},
  volume={24},
  number={2},
  pages={418-422},
  doi={10.1109/LCOMM.2019.2956526}}

@article{zhao2018exact,
  title={Exact and asymptotic ergodic capacity analysis of the hybrid satellite-terrestrial cooperative system over generalized fading channels},
  author={Zhao, Yue and Chen, Huifang and Xie, Lei and Wang, Kuang},
  journal={IET Communications},
  volume={12},
  number={11},
  pages={1342--1350},
  year={2018},
  publisher={Wiley Online Library} }

@ARTICLE{9946428,
  author={Lai, Chuan-Chi and Bhola and Tsai, Ang-Hsun and Wang, Li-Chun},
  journal={IEEE Trans. Veh. Technol.}, 
  title={Adaptive and Fair Deployment Approach to Balance Offload Traffic in Multi-{UAV} Cellular Networks}, 
  year={2023},
  volume={72},
  number={3},
  pages={3724-3738},
  doi={10.1109/TVT.2022.3221557}}

@ARTICLE{9535285,
  author={Singya, Praveen Kumar and Alouini, Mohamed-Slim},
  journal={IEEE Trans. Aerosp. Electron. Syst.}, 
  title={Performance of {UAV}-Assisted Multiuser Terrestrial-Satellite Communication System Over Mixed {FSO/RF} Channels}, 
  year={2022},
  volume={58},
  number={2},
  pages={781-796},
  doi={10.1109/TAES.2021.3111787}}

@ARTICLE{9560149,
  author={Hajijamali Arani, Atefeh and Azari, M. Mahdi and Hu, Peng and Zhu, Yeying and Yanikomeroglu, Halim and Safavi-Naeini, Safieddin},
  journal={IEEE IoT J.}, 
  title={Reinforcement Learning for Energy-Efficient Trajectory Design of {UAV}s}, 
  year={2022},
  volume={9},
  number={11},
  pages={9060-9070},
  doi={10.1109/JIOT.2021.3118322}}

@ARTICLE{6449258,
  author={Sreng, Sokchenda and Escrig, Benoit and Boucheret, Marie-Laure},
  journal={IEEE Trans. Wireless Commun.}, 
  title={Exact Symbol Error Probability of {Hybrid/Integrated} Satellite-Terrestrial Cooperative Network}, 
  year={2013},
  volume={12},
  number={3},
  pages={1310-1319},
  doi={10.1109/TWC.2013.013013.120899}}

@ARTICLE{8894851,
  author={Huang, Qingquan and Lin, Min and Zhu, Wei-Ping and Chatzinotas, Symeon and Alouini, Mohamed-Slim},
  journal={IEEE Transactions on Aerospace and Electronic Systems}, 
  title={Performance Analysis of Integrated Satellite-Terrestrial Multiantenna Relay Networks With Multiuser Scheduling}, 
  year={2020},
  volume={56},
  number={4},
  pages={2718-2731},
  doi={10.1109/TAES.2019.2952698}}

@ARTICLE{9672696,
  author={Zhao, Zhongyuan and Xu, Guanjun and Zhang, Ning and Zhang, Qinyu},
  journal={IEEE Transactions on Vehicular Technology}, 
  title={Performance Analysis of the Hybrid Satellite-Terrestrial Relay Network With Opportunistic Scheduling Over Generalized Fading Channels}, 
  year={2022},
  volume={71},
  number={3},
  pages={2914-2924},
  doi={10.1109/TVT.2021.3139885}}

@ARTICLE{9062335,
  author={Kong, Huaicong and Lin, Min and Zhu, Wei-Ping and Amindavar, Hamidreza and Alouini, Mohamed-Slim},
  journal={IEEE Wireless Commun. Lett.}, 
  title={Multiuser Scheduling for Asymmetric {FSO/RF} Links in Satellite-{UAV}-Terrestrial Networks}, 
  year={2020},
  volume={9},
  number={8},
  pages={1235-1239},
  doi={10.1109/LWC.2020.2986750}}

@ARTICLE{8387982,
  author={Qi, Ting and Feng, Wei and Wang, Youzheng},
  journal={China Communications}, 
  title={Outage performance of non-orthogonal multiple access based unmanned aerial vehicles satellite networks}, 
  year={2018},
  volume={15},
  number={5},
  pages={1-8},
  doi={10.1109/CC.2018.8387982}}

@ARTICLE{9714005,
  author={Wang, Ningyuan and Li, Feng and Chen, Dong and Liu, Liang and Bao, Zeyu},
  journal={IEEE Transactions on Vehicular Technology}, 
  title={{NOMA}-Based Energy-Efficiency Optimization for {UAV} Enabled Space-Air-Ground Integrated Relay Networks}, 
  year={2022},
  volume={71},
  number={4},
  pages={4129-4141},
  doi={10.1109/TVT.2022.3151369}}

@INPROCEEDINGS{7899525,
  author={Akiyoshi, Tomohiro and Okamoto, Eiji and Tsuji, Hiroyuki and Miura, Amane},
  booktitle={International Conference on Information Netorking}, 
  title={Performance improvement of satellite/terrestrial integrated mobile communication system using unmanned aerial vehicle cooperative communications}, 
  year={2017},
  volume={},
  number={},
  pages={417-422},
  doi={10.1109/ICOIN.2017.7899525}}

@ARTICLE{9206550,
  author={Ranjha, Ali and Kaddoum, Georges},
  journal={IEEE Internet of Things J.}, 
  title={{URLLC} Facilitated by Mobile {UAV} Relay and {RIS}: A Joint Design of Passive Beamforming, Blocklength, and {UAV} Positioning}, 
  year={2021},
  volume={8},
  number={6},
  pages={4618-4627},
  doi={10.1109/JIOT.2020.3027149}}

@INPROCEEDINGS{9149409,
  author={Kang, Zhenyu and You, Changsheng and Zhang, Rui},
  booktitle={Proc. IEEE Int. Conf. Commun. (ICC)}, 
  title={Placement Learning for Multi-{UAV} Relaying: A {G}ibbs Sampling Approach}, 
  year={2020},
  volume={},
  number={},
  pages={1-6},
  doi={10.1109/ICC40277.2020.9149409}}

@ARTICLE{8653373,
  author={Qin, Dong and Wang, Yuhao and Zhou, Fuhui and Wong, Kai-Kit},
  journal={IEEE Sys. J.}, 
  title={Performance Analysis of {AF} Relaying With Selection Combining in Nakagami-$m$ Fading}, 
  year={2019},
  volume={13},
  number={3},
  pages={2375-2385},
  doi={10.1109/JSYST.2019.2897824}}

@ARTICLE{6676778,
  author={Som, Pritam and Chockalingam, A.},
  journal={IEEE Commun. Lett.}, 
  title={Bit Error Probability Analysis of {SSK} in {DF} Relaying with Threshold-Based Best Relay Selection and Selection Combining}, 
  year={2014},
  volume={18},
  number={1},
  pages={18-21},
  doi={10.1109/LCOMM.2013.111113.131593}}

@ARTICLE{9385374,
  author={Fang, Xinran and Feng, Wei and Wei, Te and Chen, Yunfei and Ge, Ning and Wang, Cheng-Xiang},
  journal={IEEE IoT J.}, 
  title={{5G} Embraces Satellites for {6G} Ubiquitous {IoT}: Basic Models for Integrated Satellite Terrestrial Networks}, 
  year={2021},
  volume={8},
  number={18},
  pages={14399-14417},
  doi={10.1109/JIOT.2021.3068596}}

@article{al2014optimal,
  title={Optimal {LAP} altitude for maximum coverage},
  author={Al-Hourani, Akram and Kandeepan, Sithamparanathan and Lardner, Simon},
  journal={IEEE Wireless Commun. Lett.},
  volume={3},
  number={6},
  pages={569--572},
  year={2014},
  publisher={IEEE}
}

@INPROCEEDINGS{7037248,
  author={Al-Hourani, Akram and Kandeepan, Sithamparanathan and Jamalipour, Abbas},
  booktitle={ IEEE GLOBECOM}, 
  title={Modeling air-to-ground path loss for low altitude platforms in urban environments}, 
  year={2014},
  volume={},
  number={},
  pages={2898-2904},
  doi={10.1109/GLOCOM.2014.7037248}}

@ARTICLE{10164260,
  author={Nguyen, Minh-Hien T. and Bui, Tinh T. and Nguyen, Long D. and Garcia-Palacios, Emiliano and Zepernick, Hans-Jürgen and Shin, Hyundong and Duong, Trung Q.},
  journal={IEEE Transactions on Intelligent Transportation Systems}, 
  title={Real-Time Optimized Clustering and Caching for {6G} Satellite-{UAV}-Terrestrial Networks}, 
  year={2023},
  volume={},
  number={},
  pages={1-11},
  doi={10.1109/TITS.2023.3287279}}

@ARTICLE{8489991,
  author={Wang, Jingjing and Jiang, Chunxiao and Wei, Zhongxiang and Pan, Cunhua and Zhang, Haijun and Ren, Yong},
  journal={IEEE IoT J.}, 
  title={Joint {UAV} Hovering Altitude and Power Control for Space-Air-Ground {IoT} Networks}, 
  year={2019},
  volume={6},
  number={2},
  pages={1741-1753},
  doi={10.1109/JIOT.2018.2875493}}

@article{zong2019optimal,
  title={Optimal satellite {LEO} constellation design based on global coverage in one revisit time},
  author={Zong, Peng and Kohani, Saeid},
  journal={International Journal of Aerospace Engineering},
  volume={2019},
  pages={1--12},
  year={2019},
  publisher={Hindawi Limited}
}

@ARTICLE{9775682,
  author={Abderrahim, Wiem and Amin, Osama and Alouini, Mohamed-Slim and Shihada, Basem},
  journal={IEEE Trans. Commun.}, 
  title={Proactive Traffic Offloading in Dynamic Integrated Multi-satellite Terrestrial Networks}, 
  year={2022},
  volume={70},
  number={7},
  pages={4671-4686},
  doi={10.1109/TCOMM.2022.3175482}}

@INPROCEEDINGS{10464433,
  author={Omran, Allafi and Cheriet, Mohamed and Sboui, Lokman},
  booktitle={2023 IEEE Globecom Workshops (GC Wkshps)}, 
  title={{UAV}-Assisted Offloading Via Satellite Backhaul for Post-Disaster and Crowded Cellular Networks}, 
  year={2023},
  volume={},
  number={},
  pages={824-829},
  doi={10.1109/GCWkshps58843.2023.10464433}}

@ARTICLE{4062836,
  author={Chowdhury, Pulak K. and Atiquzzaman, Mohammed and Ivancic, William},
  journal={IEEE Commun. Surv. Tutor.}, 
  title={Handover schemes in satellite networks: state-of-the-art and future research directions}, 
  year={2006},
  volume={8},
  number={4},
  pages={2-14},
  doi={10.1109/COMST.2006.283818}}

@article{labrador2009approach,
  title={An Approach to Cooperative Satellite Communications in {4G} Mobile Systems.},
  author={Labrador, Yuri and Karimi, Masoumeh and Pan, Deng and Miller, Jerry},
  journal={J. Commun.},
  volume={4},
  number={10},
  pages={815--826},
  year={2009}
}

@book{hong2010cooperative,
  title={Cooperative communications and networking: technologies and system design},
  author={Hong, Y-W Peter and Huang, Wan-Jen and Kuo, C-C Jay},
  year={2010},
  publisher={Springer Science \& Business Media}
}

@inproceedings{holter2002optimal,
  title={The optimal weights of a maximum ratio combiner using an eigenfilter approach},
  author={Holter, Bengt and Oien, Geir E},
  booktitle={5th Nordic Signal Processing Symposium},
  year={2002},
  organization={Citeseer}
}

@article{jain1984quantitative,
  title={A quantitative measure of fairness and discrimination},
  author={Jain, Rajendra K and Chiu, Dah-Ming W and Hawe, William R and others},
  journal={Eastern Research Laboratory, Digital Equipment Corporation, Hudson, MA},
  volume={21},
  year={1984}
}

@ARTICLE{7070670,
  author={Sun, Ruoyu and Hong, Mingyi and Luo, Zhi-Quan},
  journal={IEEE J. Sel. Areas in Commun.}, 
  title={Joint Downlink Base Station Association and Power Control for Max-Min Fairness: Computation and Complexity}, 
  year={2015},
  volume={33},
  number={6},
  pages={1040-1054},
  doi={10.1109/JSAC.2015.2416982}}

@ARTICLE{9555387,
  author={Bhola and Lai, Chuan-Chi and Wang, Li-Chun},
  journal={IEEE Transactions on Vehicular Technology}, 
  title={The Outage-Free Replacement Problem in Unmanned Aerial Vehicle Base Stations}, 
  year={2021},
  volume={70},
  number={12},
  pages={13390-13395},
  doi={10.1109/TVT.2021.3116701}}

@ARTICLE{9099807,
  author={Zhang, Zhaoji and Li, Ying and Huang, Chongwen and Guo, Qinghua and Liu, Lei and Yuen, Chau and Guan, Yong Liang},
  journal={IEEE Internet of Things Journal}, 
  title={User Activity Detection and Channel Estimation for Grant-Free Random Access in {LEO} Satellite-Enabled Internet of Things}, 
  year={2020},
  volume={7},
  number={9},
  pages={8811-8825},
  doi={10.1109/JIOT.2020.2997336}}

@INPROCEEDINGS{6503767,
  author={Badarneh, Osamah S. and Kadoch, Michel},
  booktitle={Proc. IEEE Global Commun. Conf. (GLOBECOM)}, 
  title={Performance analysis of dual-hop systems with fixed-gain relays over generalized {$\eta$-$\mu$} fading channels},  
  year={2012},
  volume={},
  number={},
  pages={4148-4152},
  doi={10.1109/GLOCOM.2012.6503767}}

@online{laws_taiwan,
	title = {Drone Laws in {Taiwan}},
	note = {access on July. 8, 2022},
	url = { https: //uavcoach.com/drone-laws-in-taiwan/},
}

@ARTICLE{9681631,
  author={Lin, Xingqin and Cioni, Stefano and Charbit, Gilles and Chuberre, Nicolas and Hellsten, Sven and Boutillon, Jean-Francois},
  journal={IEEE Commun. Magazine }, 
  title={On the Path to {6G}: Embracing the Next Wave of Low Earth Orbit Satellite Access}, 
  year={2021},
  volume={59},
  number={12},
  pages={36-42},
  doi={10.1109/MCOM.001.2100298}}

@ARTICLE{6280241,
  author={Sanguinetti, Luca and D'Amico, Antonio A. and Rong, Yue},
  journal={IEEE J. Sel. Areas Commun.}, 
  title={A Tutorial on the Optimization of {A}mplify-and-{F}orward {MIMO} {R}elay Systems}, 
  year={2012},
  volume={30},
  number={8},
  pages={1331-1346},
  doi={10.1109/JSAC.2012.120904}}

@article{AB-GC2025,
  title={Timing advance and Doppler shift estimation in {LEO} satellite networks: A recursive Bayesian study},
  author={Balakrishnan, Ashutosh and Popineau, Pierre and Martins, Philippe},
  journal={arXiv preprint arXiv:2506.08739},
  year={2025}
}

@inproceedings{AB-VTC2024,
  title={Intelligent {UAV} Swarm Coexistence in {DTV} Bands},
  author={Dubey, Rajrshi and Balakrishnan, Ashutosh and De, Swades},
  booktitle={IEEE 100th Vehicular Technology Conference (VTC2024-Fall) },
  pages={1--5},
  year={2024},
  organization={IEEE}
}

@inproceedings{AB-ICC2024,
  title={{HAPS}-aided power grid connected green communication framework: Architecture and optimization},
  author={Balakrishnan, Ashutosh and De, Swades and Wang, Li-Chun},
  booktitle={Proc. IEEE Int. Conf. Commun. (ICC)},
  pages={2847--2852},
  year={2024},
  organization={IEEE}
}
%
\begin{IEEEbiography}[{\includegraphics[width=1in,height=1.25in,clip,keepaspectratio]{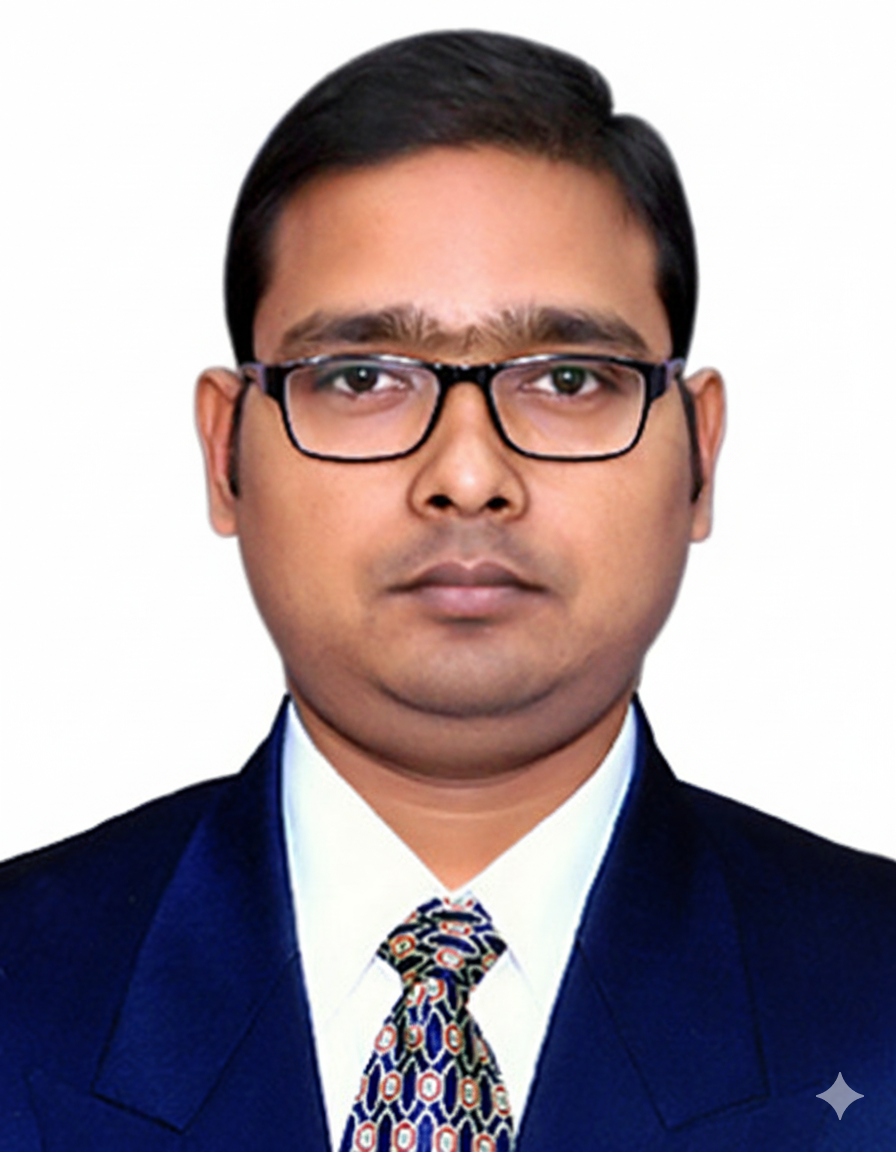}}]{Bhola}
(Member, IEEE) received the B.Tech. degree in electronics and communication engineering from the Dr. A.P.J. Abdul Kalam Technical University, Lucknow, India, in 2011, the M.Tech. degree in wireless communication and networks from the Gautam Buddha University, Greater Noida, India, and the Ph.D. degree in electrical engineering and computer science from the National Yang Ming Chiao Tung University, Hsinchu, Taiwan, in 2025. He is currently a Postdoctoral Fellow in the Institute of Communications Engineering, National Sun Yat-sen University, Kaohsiung, Taiwan, where he is affiliated with the 6G communication and sensing research center. His research interests include mobility management, UAV communication systems, non-terrestrial networks, future internet, and network optimization for 5G/6G systems. He received the Outstanding Research Award from the Department of EECS at the NYCU, Hsinchu, Taiwan in 2021. He has served as a reviewer for numerous IEEE journals and conferences, including \textsc{IEEE Transactions on Vehicular Technology}, \textsc{IEEE Transactions on Intelligent Transportation Systems}, \textsc{IEEE WoWMoM}, \textsc{IEEE VTC}, \textsc{IEEE SmartGridComm}, and others, as well as \textit{Wireless Networks} (Springer).
\end{IEEEbiography}
\begin{IEEEbiography}[{\includegraphics[width=1in,height=1.25in,clip,keepaspectratio]{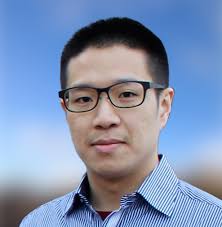}}]{Yu-Jia Chen} (Senior member, IEEE)  received the B.S. and Ph.D. degrees in electrical engineering from National Chiao Tung University, Hsinchu, Taiwan, in 2010 and 2015, respectively.
From 2015 to 2018, he was a Postdoctoral Research Fellow with National Chiao Tung University. From November 2018 to January 2019, he was a Visiting Scholar with the John A. Paulson School of Engineering and Applied Sciences, Harvard University, Cambridge, MA, USA. In 2019, he joined National Central University, Taoyuan, Taiwan, where he is currently an Associate Professor with the Department of Communication Engineering. His research interests include AI-enabled wireless networks, aerial sensor networks, and IoT security. He has published more than 50 articles in peer-reviewed international journals and conference papers.
Dr. Chen was a recipient of the Best Paper Award for Postdoctoral Research Fellow from the Ministry of Science and Technology, Taiwan, in 2018, the Outstanding Advisor Award from the Taiwan Institute of Electrical and Electronic Engineering in 2019, Outstanding Research Award of the National Central University, Taiwan, in 2023, and the Young Researcher Award from the Taiwan Association of Cloud Computing in 2023. He has been serving as a Technical Organizing Committee member and the Symposium Co-Chair for many international conferences and symposia, including Globecom, ICC, and PIMRC. He has experience with tutorials at academic conferences, such as Globecom and VTC.
\end{IEEEbiography}
\begin{IEEEbiography}[{\includegraphics[width=1in,height=1.25in,clip,keepaspectratio]{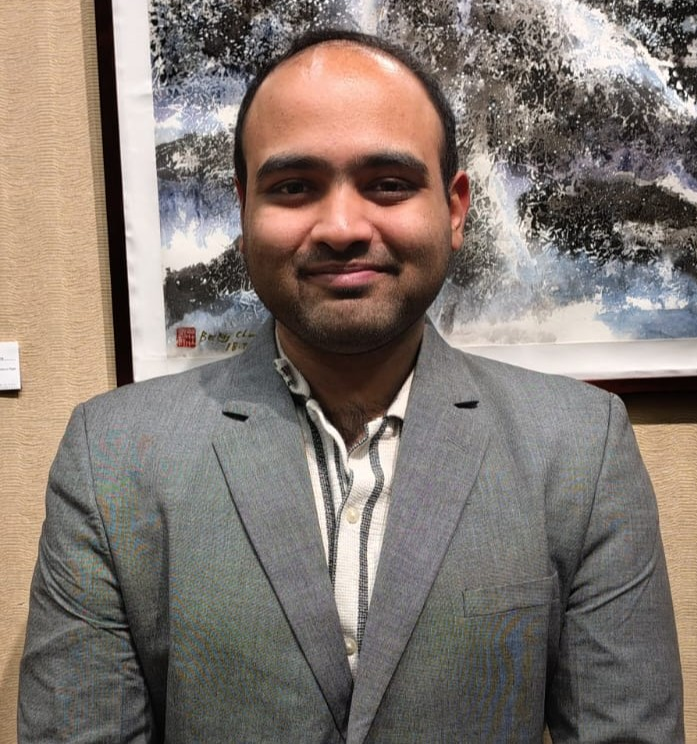}}]{Ashutosh Balakrishnan}
(Member, IEEE) received the Ph. D. degree in Electrical Engineering jointly from I.I.T. Delhi, India and N.Y.C.U. Taiwan, in July 2024; and the B. Tech. (Hons.) degree in Electronics and Telecommunication engineering from the N.I.T. Raipur, India, in 2019. 
Dr. Balakrishnan is currently a Maître de Conférences, at the Dept. of Computer Science and Networks, Télécom Paris, France. He is also an associate member of LINCS, France, and a visiting researcher at INRIA Paris. Previously (Sept. 2024 - July 2025), he was a post-doctoral fellow at Télécom Paris, France. 
His research interests broadly pertain to analytical modeling based system design and performance evaluation of stochastic networks. Current research directions include LEO satellite networks, joint communication \& sensing, and green communications. 
He has been a recipient of the Prime Minister’s Research Fellowship, Govt. of India; the Best Journal Award, ICST, NYCU Taiwan (2022); and the Innovation Excellence Award, COMSNETS (2024). Dr. Balakrishnan currently serves as the Editor of IEEE TRANSACTIONS ON COMMUNICATIONS. 
\end{IEEEbiography}
\begin{IEEEbiography}[{\includegraphics[width=1in,height=1.25in,clip,keepaspectratio]{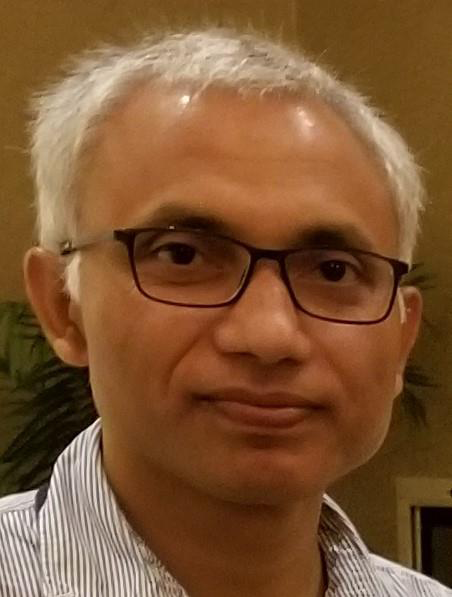}}]{Swades De}
(Senior member, IEEE)  received the B.Tech. degree in Radiophysics and Electronics from the University of Calcutta in 1993, the M.Tech. degree in Optoelectronics and Optical communication from IIT Delhi in 1998, and the Ph.D. degree in Electrical Engineering from the State University of New York at Buffalo in 2004. Dr. De is currently a Professor with the Department of Electrical Engineering, IIT Delhi. Before moving to IIT Delhi in 2007, he was a Tenure-Track Assistant Professor with the Department of ECE, New Jersey Institute of Technology, Newark, NJ, USA, from 2004–2007. He worked as an ERCIM Post-doctoral Researcher at ISTI-CNR, Pisa, Italy (2004), and has nearly five years of industry experience in India on telecom hardware and software development, during 1993–1997 and in 1999. His research interests are broadly in communication networks, with emphasis on performance modeling and analysis. Current directions include resource allocation, energy harvesting, wireless energy transfer, energy sustainable and green communications, spectrum sharing, smart grid networks, and IoT communications. Dr. De currently serves as an Editor of IEEE TRANSACTIONS ON WIRELESS COMMUNICATIONS, and Associate Editor of IEEE TRANSACTIONS ON VEHICULAR TECHNOLOGY, and IEEE WIRELESS COMMUNICATIONS MAGAZINE.
\end{IEEEbiography}
\begin{IEEEbiography}
[{\includegraphics[width=1.1in,trim=1cm 0.5cm 1cm 0.5cm,clip]{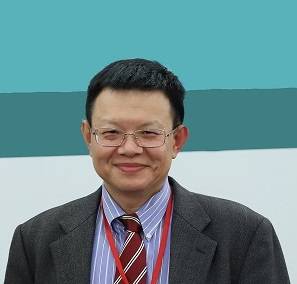}}]{Li-Chun Wang} (Fellow, IEEE) received Ph.D. degree from the Georgia Institute of Technology in 1996. From 1996 to 2000, he served as a Senior Researcher at the Wireless Communications Research Institute at AT\&T Labs. He is currently the Dean of the College of Electrical Engineering and a Lifetime Chair Professor in the Department of Electrical Engineering at National Yang Ming Chiao Tung University. Dr. Wang was elected as a Fellow of the Institute of Electrical and Electronics Engineers (IEEE) in 2011 for his contributions to the design of cellular architectures and wireless resource management in wireless networks.
Dr. Wang has received numerous awards and honors, including the Distinguished Research Awards from National Science and Technology Council twice (2012 and 2016), the Future Tech Award from the National Science and Technology Council (2021), the Chinese Institute of Engineers (CIE) Outstanding Electrical Engineering Professor Award (2022), the Outstanding Engineering Professor Award from the Chinese Institute of Electrical Engineering (2009), the K.T. Li Fellow Award (2021) and the Medal of Honor (2024) from the Institute of Information \& Computing Machinery (iICM), the Outstanding ICT. Elite Award (2020), the Y. Z. Hsu Scientific Paper Award (2013), the Y. Z. Hsu Scientific Chair Professor (2023), the Chinese Institute of Engineers (CIE) Fellow (2025), and the Pan Wen Yuan Foundation, Outstanding Research Award (2025), the Chinese Institute of Engineers (CIE) Engineering Medal (2025).
Dr. Wang has made significant contributions to the research fields of wireless communication and information technology. He has served as an IEEE Tutorial Speaker multiple times, promoting international cooperation and talent cultivation. According to Google Scholar, Dr. Wang's research works have been cited over 11,515 times with an h-index of 54. He was listed in the "2024 Annual Global Top 2\% Scientists" and "Lifetime Scientific Impact Rankings" by Stanford University. He is also ranked as a top Taiwanese international scholar in the field of computer science by the Guide2 Research website.
Dr. Wang serves as the Director of the Chunghwa Telecom-NYCU Innovation Research Center and the NYCU-IBM iIoT Research Center. He has collaborated with numerous domestic and international companies and holds 49 domestic and international patents, sixteen of which have been applied in commercial products. He is currently an Associate Editor of the IEEE Transactions on Wireless Communications and the IEEE Internet of Things Journal. His recent research interests lie in data-driven intelligent wireless communications, brain technology, and sustainable development. He has published over 300 journal and conference papers and co-edited the book "Key Technologies for 5G Wireless Communications" (Cambridge University Press, 2017).
\end{IEEEbiography}

\end{document}